\documentclass[hidelinks]{siamart220329}
\usepackage[utf8]{inputenc}
\usepackage{amssymb}
\usepackage[margin=1in]{geometry}
\usepackage{bold-extra}
\usepackage{csquotes}
\usepackage{comment}
\usepackage{tikz}
\usetikzlibrary{shapes.geometric}
\tikzset{
    itri/.style n args={2}{%
    isosceles triangle, 
    draw, 
    shape border rotate = 90,
    minimum height = 5mm,
    minimum width = 4mm,
    inner sep = 2pt,
  }
}
\tikzset{
    itri_dashed/.style n args={2}{%
    isosceles triangle, 
    draw, dashed
    shape border rotate = 90,
  }  
}

\usepackage{hyperref}
\usepackage{nameref}
\newtheorem{remark}{Remark}

\newcommand{\BTnstar}{\mathcal{BT}^\ast_n}

\newcommand{\BTn}{\mathcal{BT}_n}
\newcommand{\Tfb}{T^{\mathit{fb}}_h}
\newcommand{\Tmb}{T^{\mathit{mb}}_n}
\newcommand{\Tcat}{T^{\mathit{cat}}_n}

\newcommand{\Tgfb}{T^{\mathit{gfb}}_n}
\newcommand{\Path}{\mathcal{P}}
\newcommand{\Pone}{\mathcal{P}_1}
\newcommand{\Ptwo}{\mathcal{P}_2}

\usepackage{xcolor}

\headers{The GFB Tree and Tree Imbalance Indices }
{Cleary, Fischer, and St.~John}

\author{Sean Cleary\thanks{Department of Mathematics, The City College of New York, NY 10031, USA
  (\email{scleary@ccny.cuny.edu}). 
}
\and Mareike Fischer\thanks{University of Greifswald, Institute of  Mathematics and Computer Science, Walther-Rathenau-Str. 47, 17487 Greifswald, Germany
  (\email{email@mareikefischer.de, mareike.fischer@uni-greifswald.de}).}
\and Katherine St.~John\thanks{Department of Computer Science, Hunter College of CUNY and the Division of Invertebrate Zoology, American Museum of Natural History, New York, NY }
(\email{stjohn@hunter.cuny.edu}).}

\title{The GFB Tree and Tree Imbalance Indices}

\begin{document}

\maketitle

\begin{abstract}
    Tree balance plays an important role in various research areas in phylogenetics and computer science. Typically, it is measured with the help of a balance index or imbalance index. There are more than 25 such indices available, recently surveyed in a book by Fischer et al. They are used to rank rooted binary trees on a scale from the most balanced to the least balanced. 
    We show that a wide range of subtree-size based measures satisfying concavity and monotonicity conditions are minimized by the complete or greedy-from-the-bottom (GFB) tree and maximized by the caterpillar tree, yielding an  infinitely large family of distinct new imbalance indices.
    Answering an open question from the literature, we show that one such established measure, the $\widehat{s}$-shape statistic, has the GFB tree as its unique minimizer. We also provide an alternative characterization of GFB trees, showing that they are equivalent to complete trees, which arise in different contexts.  We give asymptotic bounds on the expected $\widehat{s}$-shape statistic under the uniform and Yule-Harding distributions of trees, and answer questions for the related $Q$-shape statistic as well.
\end{abstract}

\section{Introduction}

Trees are a canonical data structure, providing an efficient way to implement fundamental concepts such as dynamic sets as well as representing hierarchical and phylogenetic relationships between data (see \cite{clr} and  \cite{Semple2003}).  
Much of the power of the tree data structure relies on well-distributed branching that can yield 
tree height logarithmic in the total size of the tree, and result in efficient access, assuming a reasonable balance. 
The balance, or lack thereof, often affects the running time of algorithms, with many tree-based algorithms having significantly different times depending if they are very balanced (an element in a balanced, binary search tree on $n$ leaves can be found in $O(\log n)$ time) or very imbalanced (the same search has $O(n)$ time for the pectinate or caterpillar tree) \cite{clr}. 
There are many different measures suggested to assess balance, this fundamental property of trees, surveyed  in \cite{fischerBook}. While similar in format, these indices can yield quite different rankings of trees, as illustrated in Figure~\ref{fig:plotn10}. This figure compares the rankings of three different indices considered in the present manuscript (the $\widehat{s}$-shape statistic, the $Q$-shape statistic and the Sackin index) as well as two other well-known indices, namely the popular Colless and total cophenetic indices for the case $n=10$. 

When the number of leaves, $n$, is a power of 2, with $n=2^h$ for some $h$, all imbalance indices are minimal for the fully balanced tree and maximal for the caterpillar tree.  For trees which are neither the fully balanced tree nor caterpillars, the
values of each of these indices lie in the interval containing these extremes for that  imbalance index.

\begin{figure}[t]
  \centering
  \includegraphics[scale=1]{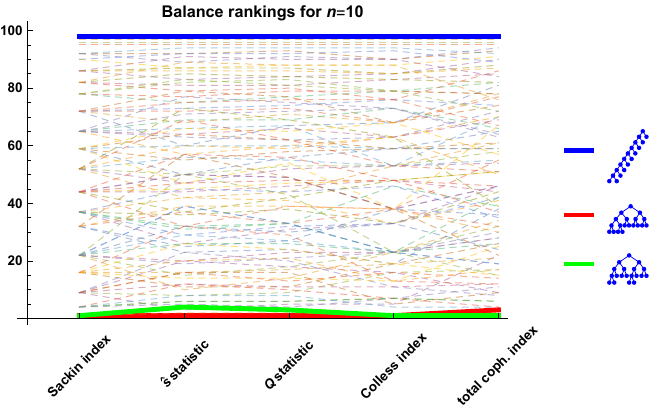}
  \caption{The rankings of all 98 tree shapes of size 10 with respect to the Sackin index, the $\widehat{s}$-shape statistic, the $Q$-shape statistic, the Colless index, and  the total cophenetic index (see \cite{fischerBook} for a survey of indices).  All indices rate the caterpillar tree shown in blue as extremely unbalanced, and  the indices rank trees of intermediate balance in
  different orders. The sets of other minimal trees with respect to the other rankings contain the maximally balanced tree shown in green and/or the GFB tree shown in red.}
  \label{fig:plotn10}
\end{figure}

We show a general result that applies to a broad range of tree imbalance indices.  If a tree imbalance index function using subtree sizes satisfies some concavity and increasing conditions, the minimum value is achieved by \enquote{greedy from the bottom} (GFB) trees, as named in \cite{fischerBook}. We will show that these trees coincide with trees termed  \enquote{complete trees}  \cite{Fill1996}, which have occurred in other contexts.  The $\widehat{s}$-shape statistic of Blum and Fran{\c{c}}ois \cite{blum2006imbalance} satisfies this property.  The $\widehat{s}$-shape statistic
sums the logarithms of the subtree sizes across the tree, so for a rooted binary tree $T$, the  $\widehat{s}$-shape statistic is $\sum \log (n_v-1)$, where $n_v$ is the number of leaves in the subtree rooted at the internal node $v$.  Blum and Fran{\c{c}}ois \cite{blum2006imbalance} note that once a normalizing constant has been removed, the $\widehat{s}$-shape statistic corresponds to the logarithm of the probability of a tree in the uniform or equal rates model (ERM) for generating random trees (see  \cite[p.~29-30]{Semple2003}), and provides a strong tool rejecting the various tree models  against the Yule-Harding or proportional to distinguished arrangements model (PDA) (see \cite{kersting2024} and Section~\ref{sec:prelim} of the present manuscript concerning probabilistic models of phylogenetic trees).  

Our findings answer the open question of which trees, among all rooted binary trees of size $n$, minimize the  $\widehat{s}$-shape statistic as well as the question if there is an explicit formula for the minimum value of $\widehat{s}$, both of which were posed in \cite{fischerBook}.  
We further analyze a related measure of tree balance,  namely the $Q$-shape statistic, described by Fill \cite{Fill1996}. Motivated by building binary search trees from random permutations, it  can be recast as a close parallel to the $\widehat{s}$-shape statistic. We discuss how his work shows the maximal and minimal tree shapes for this statistic and the moments for it under the uniform distribution.  We consider the distribution of the relevant statistic under the Yule-Harding distribution as well.
We describe many indices built from
concave functions of subtree sizes as imbalance indices, in that they have maximal values on caterpillar trees and minimal values on GFB trees, which necessarily include fully balanced trees when the size is a power of two. We show that there are infinitely many distinct such indices.

The present manuscript is organized as follows: In Section \ref{sec:prelim}, we present all definitions and notations needed throughout the manuscript. In Section \ref{sec:prior}, we state some known results from the literature which we will use to derive our results. Section \ref{sec:results} then contains all our results, which are structured as follows: Section \ref{sec:gfbmin} gives an overview of the minimizing properties of the GFB tree. The results of this subsection are used in Section \ref{sec:bal} to show that both the $\widehat{s}$-shape statistic and the $Q$-shape statistic belong to an infinitely large family of different imbalance indices. Section \ref{sec:explicitformulas} then gives two explicit formulas for the minimum value both for $\widehat{s}$ and $Q$, both of which are derived from the GFB tree. Finally, Section \ref{sec:expected} derives some expected values for the $\widehat{s}$-shape statistic, which also answers open questions posed in \cite{fischerBook}. We conclude with a brief discussion in Section \ref{sec:discussion}.

\section{Definitions} We outline the terminology used, following the standard notions from \cite{fischerBook,king2021,Steel2016}.  
\label{sec:prelim}

\subsubsection*{Graph theoretical trees and phylogenetic trees} 

A \emph{rooted binary tree}, or simply a \emph{tree}, is a directed graph $T=(V(T),E(T))$ with vertex set $V(T)$ and edge set $E(T)$, containing precisely one vertex of in-degree 0, the \emph{root} (denoted by $\rho$), such that for every $v \in V(T)$ there exists a unique path from $\rho$ to $v$, and such that all vertices have out-degree $0$ or $2$. In particular, the edges are directed away from the root. Nodes with out-degree 2 are {\em internal nodes} and nodes with out-degree 0 are {\em leaf nodes} or {\em leaves}. We use $\mathring{V}(T)$ (or simply $\mathring{V}$) to denote the set of internal vertices of $T$ and $V_L(T)$ to denote the set of leaves of $T$, respectively. 

Tree balance is independent of any leaf labeling, but in phylogenetics, the leaf labeling plays an important role and is used when considering evolutionary models. A \emph{rooted binary phylogenetic $X$-tree} $\mathcal{T}$ (or simply \emph{phylogenetic tree}) is simply a tuple $\mathcal{T}=(T, \phi)$, where $T$ is a rooted binary tree and $\phi$ is a bijection from the set of leaves $V_L(T)$ to $X$. The (unlabeled) tree $T$ is often referred to as the \emph{topology} or \emph{tree shape} of $\mathcal{T}$ and $X$ is called the \emph{taxon set} of $\mathcal{T}$. We assume that the label sets are the numbers: $X=\{1, \ldots, n\}$.

We consider two trees $T$ and $T'$ as equal if they are isomorphic; that is, if there exists a mapping $\theta:V(T)\rightarrow V(T')$ such that for all $u,v \in V(T)$ we have $(u,v)\in E(T) \Leftrightarrow (\theta(u),\theta(v)) \in E(T')$ and with $\theta(\rho(T))=\rho(T')$. In particular, $\theta$ is a graph isomorphism which preserves the root position. We use  $\BTnstar$ to denote the space (of isomorphism classes) of (rooted binary) trees with $n$ leaves, which are unlabeled.

We use $\BTn$ to denote the space (of isomorphism classes) of rooted binary phylogenetic $X$-trees with $\vert X \vert =n$ where the leaves are labeled. 
Moreover, we recall that $\vert \mathcal{BT}_1 \vert =1$  and $\vert \BTn \vert =(2n-3)!!=(2n-3)(2n-5)\cdots1$ for $n \geq 2$ (\cite[Corollary 2.2.4]{Semple2003}).

\begin{figure}[t]
    \centering  
    \begin{tabular}{ccccc} 
        \\
            \includegraphics[height=1in]{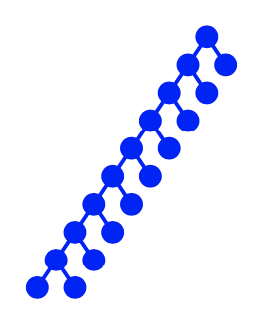} & &
            \includegraphics[height=1in]{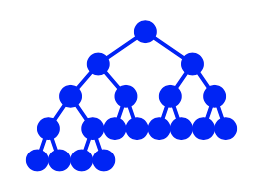} & &
            \includegraphics[height=1in]{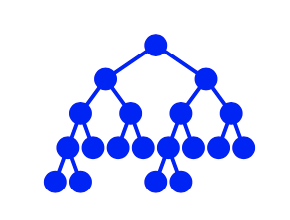} 
        \\
        (a) caterpillar && (b) greedy-from-the-bottom && (c) maximally-balanced\\
        \\
    \end{tabular}
    \caption{Three tree shapes on $10$ leaves:  
    (a) a caterpillar with 
    $\widehat{s}$-shape of $\sum_{i=2}^9 \log i = \log(9!)  
    ~ \sim 12.8 $, 
    (b) a greedy-from-the-bottom (GFB) tree with 
    $\widehat{s}$-shape of $\log(9 \cdot 5^1 \cdot 3^2) 
    \sim 6.0$, and
    (c) a maximally-balanced  tree with 
    $\widehat{s}$-shape of $\log(9 \cdot 4^2 \cdot 2^2) 
    \sim 6.4$, with the middle GFB tree  (b) being the
    unique minimizer for the $\widehat{s}$-shape statistic among all trees with 10 leaves.
    \label{fig:trees}}
\end{figure}

\subsubsection*{Vertices and subtrees}
We now define some properties of vertices and subtrees, which apply both for trees as well as phylogenetic trees. 
Throughout, for a given tree $T$, we denote the number of leaves by $n$, with $n=|V_L(T)|$. In the special case of $n=1$, the tree consists of only one vertex, which is at the same time the root and its only leaf. The \emph{size} of a tree $T$ is defined as the number of leaf nodes and is sometimes also  indicated as $|T|$.

Whenever there is a directed path $\Path$ from a vertex $u$ to a vertex $v$ in $T$, we call $u$ an \emph{ancestor} of $v$ and $v$ a \emph{descendant} of $u$. The \emph{depth} $\delta_T(v)$ 
of a vertex $v$ in $T$ denotes the number of edges on the unique path $\Path$ from the root $\rho$ of $T$ to $v$. 
If $\Path$ consists of a single edge, $u$ is the \emph{parent} of $v$ and $v$ is the \emph{child} of $u$.  If two leaves $u$ and $v$ have the same parent $w$, $u$ and $v$ form a {\em sibling pair} or \emph{cherry}, which we denote by $[u,v]$.

Every internal vertex $v$ of $T$ induces a  \emph{pending subtree} $T_v$, that is, the subtree of $T$ which has $v$ as its root and contains all descendants of $v$. The number of leaves of this subtree will be referred to as the \emph{size} of $T_v$ and will be denoted by $n_v$. Note that $n_{\rho}=n$. An internal vertex $u$ with children $v$ and $w$ is called \emph{balanced} if $|n_v-n_w|\leq 1$. Finally, for a (phylogenetic) tree $T$, we denote its \emph{standard decomposition}  into its maximal pending subtrees (that is, the two subtrees rooted at the children of the root $\rho$ of $T$) $T_a$ and $T_b$ by $T=(T_a,T_b)$. Note that the \emph{height} $h(T)$ of a tree $T$ is defined as 
$h(T)=\max_{v \in V(T)} \delta_T(v)$; that is, the height of a trees coincides with the maximum depth of its vertices.

\subsubsection*{Special trees}

We describe some trees which play important roles in tree balance. The first tree is the  \emph{caterpillar tree}, denoted by $\Tcat$, which is a rooted binary tree with $n$ leaves which either consists of only one vertex or it contains precisely one cherry (see Figure~\ref{fig:trees}(a)).  

In considering a tree of size $2^h$, the \emph{fully balanced tree of height $h$}, $\Tfb$, is the tree where every node has two children and all 
leaves have depth exactly $h$. Note that for $n \geq 2$ both maximal pending subtrees of a fully balanced tree are again fully balanced trees, and we have $\Tfb=\left(T^{\mathit{fb}}_{h-1}, T^{\mathit{fb}}_{h-1}\right)$.
The \emph{maximally balanced (MB) tree}  with $n$ leaves, denoted by $\Tmb$, is the unique rooted binary tree with $n$ leaves in which all internal vertices are balanced. Recursively, a rooted binary tree with $n \geq 2$ leaves is maximally balanced if its root is balanced and its two maximal pending subtrees are maximally balanced, with $\Tmb = \left(T_{\lceil n/2 \rceil}^{\mathit{mb}}, T_{\lfloor n/2 \rfloor}^{\mathit{mb}}\right)$ (see Figure~\ref{fig:trees}(c)). 

The \emph{greedy from the bottom (GFB) tree} 
for $n$ leaves, denoted by $\Tgfb$, is a rooted binary tree with $n$ leaves that results from greedily clustering trees of minimal sizes from an initial forest in the following manner.  We start with a forest of $n$ trees each consisting of a single vertex and proceed by successively joining two of the smallest remaining trees in the forest until the forest is resolved into a single tree, with the resulting tree of the shape depicted in Figure~\ref{fig:trees}(b).  Coronado {\em et al.} describe this construction in \cite[Algorithm 2]{Coronado2020a}. 

The \emph{complete tree} (as defined in Fill \cite{Fill1996}) for $n$ leaves, denoted by $T^c_n$, is a rooted binary tree with $n$ leaves that results from creating the fully balanced tree of size $2^{\lfloor \log_2(n)\rfloor}$, the largest fully balanced tree with $n$ leaves or fewer, ordering the leaves from the left to the right and then attaching sibling pairs on the leaves from left to right until a total of $n$ leaves are obtained. We will show below that the complete tree coincides with the GFB tree for all sizes, cf. Lemma \ref{lem:Tgfb_altDef}. In order to simplify notation, throughout we use the shorthand $\log$ instead of $\log_2$ whenever we refer to a base-2 logarithm.

Note that if $n=2^h$, we have $\Tgfb=\Tfb=\Tmb$. The latter equality follows from the fact that in the case where the number of leaves is a power of 2, $\Tfb$ is the unique tree all of whose internal vertices have a balance value of zero. The first equality holds, because if $n=2^h$, $n\mod 2^i\equiv 0$ for all $i=0,\ldots,h-1$ and during the greedy clustering procedure, every tree with $2^i$ leaves clusters with another tree of $2^i$ leaves for all $i=0,\ldots,h-1$.  This process continues until the single remaining tree is the fully-balanced tree.

\subsubsection*{Tree shape statistics and (im)balance indices} 

We focus on \emph{tree imbalance indices}. Following \cite{fischerBook}, we call a function $t:\mathcal{BT}_n^\ast\rightarrow\mathbb{R}$ a \textit{tree shape statistic (TSS)} if $t(T)$ depends only on the shape of $T$ and not on the labeling of vertices or the lengths of edges. Such a tree shape statistic $t$ is called an \emph{imbalance index} if
\begin{itemize}
    \item the caterpillar tree\index{tree!caterpillar} $\Tcat$ is the unique tree maximizing $t$ on $\BTnstar$ for all $n\geq 1$, and
    \item the fully balanced tree\index{tree!fully balanced} $\Tfb$ is the unique tree minimizing $t$ on $\BTnstar$ for all $n=2^h$  for natural $h$.
\end{itemize}

The choice of base for the logarithm only affects the following indices by a multiplicative factor, and so we presume base 2 for all evaluations below.
We first focus on several well-known and frequently used tree imbalance indices that are defined in terms of leaf counts of subtrees:

\begin{definition}\label{def:TSS} Let $T\in \BTnstar$. 
  \begin{itemize}
    \itemsep 7pt
    \item \emph{Sackin index} (cf. \cite{Fischer2021,fischerBook,Shao1990}):  
    $S(T) =\sum\limits_{v\in\mathring{V}(T)} n_v.$ 
    
    \item $\widehat{s}$-shape statistic  \cite{blum2006imbalance,fischerBook}: 
    $\widehat{s}(T) := \sum\limits_{v\in\mathring{V}(T)} \log(n_v-1) =  \log \left( \prod \limits_{v\in\mathring{V}(T)} (n_v-1) \right).$ 
    
    \item $Q$-shape statistic (related to Fill \cite{Fill1996}): 
    $Q(T) := \sum\limits_{v\in\mathring{V}(T)} \log(n_v) =  \log \left( \prod \limits_{v\in\mathring{V}(T)} n_v \right).$

    \end{itemize}
\end{definition}

Note that the logarithm base was originally
not stated \cite{blum2006imbalance}, however, it is common to use base 2 in binary (phylogenetic) trees, and we follow this convention here just as the authors of \cite{fischerBook}. In fact, we only consider logarithms of base 2 here; this is not limited to Definition \ref{def:TSS}.

Moreover, while we are mainly concerned with two imbalance indices, the $\widehat{s}$-shape statistic and the $Q$-shape statistic, we also consider broader families containing both indices.
The $Q$-shape statistic is related to measures introduced by Fill \cite{Fill1996} and differs from the $\widehat{s}$-shape statistic by a difference of one before taking the logarithms for each term in the product.  Despite the similarity, the statistics yield different rankings (see  Figures 
\ref{fig:plotn10} and \ref{fig:n10cvariable}).
Note that Fill  \cite{Fill1996} uses $Q(T)$ for the reciprocal of the quantity of which we are taking the logarithm,
and defines $L_n= -\ln(Q(T))$ for the negative of its natural logarithm.  $L_n(T)$  differs from how we choose to define the $Q$-shape statistic here in the base of logarithm, but we proceed as above to make the $Q$-shape statistic more closely parallel to the $\widehat{s}$-shape statistic.
Fill \cite{Fill1996} shows that the complete (GFB) tree minimizes the $Q$-shape statistic and the caterpillar tree maximizes it, and further computes the moments under the uniform distribution as well as the random permutation model.  King and Rosenberg \cite{king2021} employ a parallel structure to Fill's methods for similar results on the Sackin index.

\subsubsection*{Probabilistic models of phylogenetic trees} 
We consider two popular models of evolution:  the Yule-Harding and the uniform models.  
The \emph{Yule-Harding model} is a pure birth process in which species are born but do not go extinct. It is a forward process generating a tree $T$ as follows. The process starts with a single vertex and, at each step, chooses a leaf uniformly at random from those present and subsequently replaces that leaf by a cherry. As soon as the desired number $n$ of leaves is reached, leaf labels $X=\{1,\ldots,n\}$ are assigned uniformly at random to the leaves. The probability $P_{Y,n}(T)$ of generating a phylogenetic $X$-tree $\mathcal{T}=(T,\phi)$ under the Yule-Harding model is then given by Steel \cite[Proposition 3.2]{Steel2016}:
\begin{align}
     P_{Y,n}(T) &= \frac{2^{n-1}}{n!} \cdot \prod\limits_{v \in \mathring{V}(T)} \frac{1}{n_v-1}.
\end{align}

The \emph{uniform model} selects a phylogenetic $X$-tree uniformly at random from the set of all possible phylogenetic trees (\cite{Rosen1978}). 
As $\vert \mathcal{BT}_n \vert=(2n-3)!!$ for every $n \geq 1$ (with the convention that $(-1)!!=1$; see, for instance, \cite[Corollary 2.2.4]{Semple2003}), the probability $P_{U,n}(T)$ of generating a phylogenetic $X$-tree $T$ under the uniform model is thus given by 
\begin{align}
     P_{U,n}(T) &= \frac{1}{(2n-3)!!}.
\end{align}

\section{Prior results}\label{sec:prior} In order to investigate the GFB tree and its relevance for the $\widehat{s}$-shape statistic more deeply, we use the following earlier results. 

\begin{lemma}[Lemma 5 in \cite{Coronado2020a}]\label{lem:GFBmaxsubtrees} If $T = (T_a, T_b)$ is a GFB tree, then $T_a$ and $T_b$ are also GFB trees.
\end{lemma}

Note that recursively applying Lemma \ref{lem:GFBmaxsubtrees} to the maximum pending subtrees of $T$ and again their maximum pending subtrees and so forth, we easily derive the following corollary.

\begin{corollary}\label{cor:GFBsubtrees} If $T$ is a GFB tree and $v$ is a vertex of $T$ inducing the pending subtree $T_v$, then $T_v$ is also a GFB tree.
\end{corollary}

The following proposition, which has been adapted from Proposition 5 in \cite{Coronado2020a},  characterizes the sizes of the maximal pending subtrees of GFB trees.

\begin{proposition}[adapted from Proposition 5 in \cite{Coronado2020a}]\label{prop:coronado2cases} For $n \geq 2$, we let $\Tgfb=(T_a,T_b)$, where $n_a$ and $n_b$ denote the sizes of $T_a$ and $T_b$, respectively. Let $\ell_{n} =\lfloor\log(n)\rfloor$. Then, we have:
\begin{enumerate}
\item If $2^{\ell_n} \leq n \leq 3\cdot 2^{{\ell_{n}}-1}$, then $n_a=n-2^{\ell_{n}-1}$, $n_b=2^{{\ell_{n}}-1}$ and $T_b=T_{\ell_n-1}^{fb}$.
\item If $3 \cdot 2^{{\ell_{n}}-1}\leq n<2^{\ell_{n}+1}-1$, then $ n_a=2^{\ell_{n}}$, $n_b=n-2^{\ell_{n}}$ and $T_a=T_{\ell_n}^{fb}$.
\end{enumerate}
\end{proposition}

We mainly work with $k_n=\lceil\log(n)\rceil$ rather than $\ell_{n} =\lfloor\log(n)\rfloor$ to simplify some proofs later.

\begin{corollary}
\label{cor:coronado2cases}
For $n \geq 2$, we let $\Tgfb=(T_a,T_b)$, where $n_a$ and $n_b$ denote the sizes of $T_a$ and $T_b$, respectively. Let $k_n=\lceil\log(n)\rceil$. Then, we have:
\begin{enumerate}
\item If $2^{k_n-1} < n \leq 3\cdot 2^{k_{n}-2}$, then $n_a=n-2^{k_n-2}$, $n_b=2^{k_n-2}$ and $T_b=T_{k_n-2}^{fb}$.
\item If $3 \cdot 2^{{k_{n}}-2}\leq n\leq 2^{k_{n}}$, then $ n_a=2^{k_{n}-1}$, $n_b=n-2^{k_{n}-1}$ and $T_a=T_{k_n-1}^{fb}$.
\end{enumerate}
\end{corollary}

\begin{proof} For all cases in which $n$ is not a power of 2, we have $k_n=\ell_n+1$.  
For these cases, substituting $\ell_n$ in Proposition \ref{prop:coronado2cases} by $k_n-1$ directly leads to the required claims. 

So it only remains to consider the case in which $n$ is a power of 2. Note that in this case, $k_n=\lceil \log(n) \rceil =\lfloor \log(n) \rfloor=\ell_n$, so we have $n=2^{\ell_n}=2^{k_n}$. Proposition \ref{prop:coronado2cases} covers this case in Case 1, which says that $n_a=n-2^{\ell_n-1}$ and $n_b=2^{\ell_n-1}$. Using $n_a+n_b=n$ and $k_n=\ell_n$ and $n=2^{k_n}$ in this case, we get: $n_a=n-2^{\ell_n-1}=2^{k_n}-2^{k_n-1}=2^{k_n-1}$ and $n_b=2^{\ell_n-1}=2^{k_n-1}$.

 Corollary  \ref{cor:coronado2cases} covers this case in Case 2, which implies that $n_a=2^{k_n-1}$ and $n_b=2^{\ell_n-1}=2^{k_n-1}$.

\end{proof}

\section{Results} \label{sec:results}

We show that the GFB tree plays a fundamental role for the $\widehat{s}$-shape statistic, even more so than it does for other balance indices like the Sackin index, for which it is known to be contained in the set of minimal trees (cf. \cite{Fischer2021,fischerBook}). Note that the GFB tree plays a similar role for other imbalance indices like the well-known  \emph{Colless index}, for which it is also known to be minimal (\cite{Coronado2020a,fischerBook}). Figure \ref{fig:plotn10} depicts the role of the GFB tree for various indices. Notably, there are indices like the so-called total cophenetic index, which do not assume their minimal values at the GFB tree, but at the MB tree instead \cite{Mir2013}, cf. Figure \ref{fig:plotn10}.

We now start to investigate the GFB tree further. 

\subsection{Minimizing properties of the GFB tree}\label{sec:gfbmin} 
The aim of this section is to show that $T_n^{gfb}$ is the unique minimizer of all functions of the form $\Phi_f$, which we define for any rooted binary tree $T$ as follows: 

\begin{align*}\Phi_f(T)&= \sum\limits_{v \in \mathring{V}(T)} f(n_v), \end{align*}
where $f$ is any strictly monotonically increasing and strictly concave function.
Moreover, we will show that this implies that the GFB tree is the unique minimizer of the \emph{product function}, which we define as follows:  \begin{align*}\pi_c(T)&= \prod\limits_{v \in \mathring{V}(T)} (n_v+c), \end{align*}
where $T$ is a rooted binary tree and $c \in \mathbb{R}_{>-2}$ is a constant. Note that the choice of $c>-2$ guarantees that $n_v+c>0$ for all $v\in \mathring{V}(T)$ (as the smallest pending subtree size is  $n_v=2$, when $v$ is the parent of a cherry). The fact that all factors of the product function are strictly larger than 0 leads to meaningful properties of the product including the existence of the logarithm of the product function.  

In Section \ref{sec:bal}, we will subsequently show that the above minimizing properties have significant implications on tree balance, as they lead to two new families of tree (im)balance indices, one of which can be shown to be a subfamily of the other one. Moreover, our results will lead to answers to open questions concerning existing imbalance indices. 

We start with the following theorem, parts of which are based on the ideas of \cite[Theorem 4]{FischerLiebscher2021}.

\begin{theorem}\label{thm:minphif} Let $n\in \mathbb{N} $ with $n \geq 2$ and let $f:\mathbb{R}_{\geq 2}\rightarrow \mathbb{R}$ a strictly monotonically increasing and strictly concave function. That is, we have $f(n_1)>f(n_2)$ if and only if $n_1>n_2$, and we also have $f(\lambda x + (1-\lambda) y) > \lambda f(x) + (1-\lambda) f(y)$ for all $\lambda \in (0,1)$ and all $x, \ y \in \mathbb{R}_{\geq 2}$ with $x\neq y$. We consider $\Phi_f(T)= \sum\limits_{v \in \mathring{V}(T)} f(n_v)$. Then, $T_n^{gfb}$ is the unique tree in $\BTnstar$ minimizing $\Phi_f$.  
\end{theorem}

\begin{proof} Towards a contradiction, we assume $n$ is the smallest number where the minimizing tree is not the GFB tree. We let $T$ be a rooted binary tree with $n$ leaves minimizing $\Phi_f$, such that $T$ is not a GFB tree. Since there is only one tree when $n=1$, we can assume $n>1$.

By Lemma~\ref{lem:GFBmaxsubtrees}, the subtrees of a GFB tree are also GFB trees.  $T$ is by assumption not a GFB tree, so there will be at least two pending subtrees of $T$ that would form a GFB tree but do not have a common parent in $T$. This is due to the fact that $T$, just like every rooted binary tree, can be obtained from clustering, starting with a forest of one-leaf trees and clustering two trees at a time until a single tree is obtained. We do this to build $T$ by using the two smallest available trees at any point in time (just as in the GFB construction), until no such further clustering is possible. So there must be two trees which are present both in $T$ and $\Tgfb$,  that have a common parent in $\Tgfb$ but do not in $T$.  Let $T_a$ and $T_b$ be the smallest such subtrees of $T$. 

We note that as all previous trees of $T_n^{gfb}$ were also formed when building $T$, only two situations are possible: 

\begin{enumerate}
\item One of the two trees $T_a$, $T_b$, say $T_a$, is contained in $T$ as a sibling tree to some tree $T_c$, and the other one, say $T_b$, is a sibling to a subtree of $T$ containing the first tree, say $T_a$. This situation is depicted in Figure \ref{fig:case1} at the left-hand side.
\item The tree $T_a$ is a sibling to some tree $T_c$ not containing $T_b$, and $T_b$ is a sibling to some tree $T_d$ not containing $T_a$. This situation is depicted in Figure \ref{fig:case2} at the right-hand side.
\end{enumerate}

\par \vspace{0.2cm} 

We now consider these two cases separately.

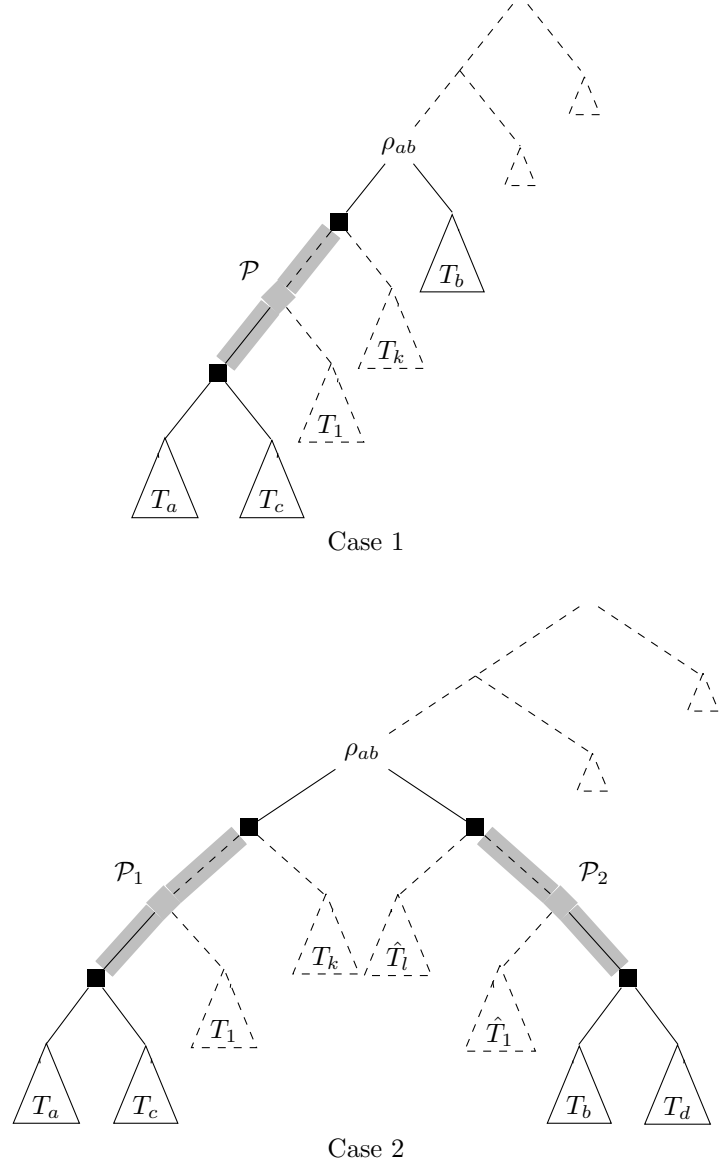
\begin{figure}
  \centering
  \begin{tabular}{c}
    \begin{tikzpicture}
    [
    level distance=10mm,
    level/.style={sibling distance=16mm}
    ]

    \node { } 
        child { edge from parent[dashed]
            child { node {$\rho_{ab}$} edge from parent[dashed]
                child{ node [fill=black,name=n1]{ } edge from parent[solid]
            child{ node [diamond,fill = lightgray,label=above left:$\Path$]{ } edge from parent [dashed,draw,
                    preaction={
                    draw,lightgray,-,
                    double=lightgray,
                    double distance=20\pgflinewidth,
                    }] 
                child{ node [fill = black]{ } edge from parent [draw,solid,
                    preaction={
                    draw,solid,lightgray,-,
                    double=lightgray,
                    double distance=20\pgflinewidth,
                    }]
                    child{ node { } 
                        child{ node[itri, xshift=1mm,yshift=3mm]{$T_{a}$} } }
                    child{ node { } 
                        child{ node[itri, xshift=-.85mm,yshift=3mm]{$T_{c}$} } }
                }                
                child{ node { } edge from parent [dashed] 
                    child{ node[itri, xshift=-1mm, yshift=3mm,dashed]{$T_{1}$} } 
                }
            }
            child{ node { } edge from parent  [dashed] 
                child{ node[itri, xshift=-1mm, yshift=2.75mm,dashed]{$T_{k}$} } 
            }
        }
                child{ node { } edge from parent[solid]
                child{ node[itri, xshift=-1mm, yshift=2.875mm,solid]{$T_{b}$} }    
            }
            }
            child { edge from parent[dashed] 
                node[itri, xshift=4mm, yshift=-9mm, dashed]{ } 
            }
        }
        child { edge from parent[dashed]
            node[itri, xshift=4mm, yshift=-9mm, dashed]{ } 
        };
    \end{tikzpicture}
    \\
    Case 1 \\
    \\
    \begin{tikzpicture}
    [
    level distance=10mm,
    level/.style={sibling distance={9cm/max(3,#1)}},
    ]

    \node { } 
        child { edge from parent[dashed]
            child { node {$\rho_{ab}$} edge from parent[dashed]
                child{ node [fill=black,name=n1]{ } edge from parent[solid]
            child{ node [diamond,fill = lightgray,label=above left:$\Pone$]{ } edge from parent [dashed,draw,
                    preaction={draw,lightgray,-,
                    double=lightgray,
                    double distance=20\pgflinewidth,
                    }] 
                child{ node [fill = black]{ } edge from parent [draw,solid,
                    preaction={
                    draw,solid,lightgray,-,
                    double=lightgray,
                    double distance=20\pgflinewidth,
                    }]
                    child{ node { } 
                        child{ node[itri, xshift=1mm,yshift=3mm]{$T_{a}$} } }
                    child{ node { } 
                        child{ node[itri, xshift=-.85mm,yshift=3mm]{$T_{c}$} } }
                }                
                child{ node { } edge from parent [dashed] 
                    child{ node[itri, xshift=-1mm, yshift=3mm,dashed]{$T_{1}$} } 
                }
            }
            child{ node { } edge from parent  [dashed] 
                child{ node[itri, xshift=-1mm, yshift=2.75mm,dashed]{$T_{k}$} } 
            }
        }
            child{ 
                node[fill=black] { } edge from parent[solid]
                child{ node { } edge from parent[dashed]
                child{ node[itri, xshift=1mm, yshift=3mm,dashed]{$\hat{T}_{l}$}
                edge from parent[dashed]
                }}        
            child{ 
                node[diamond,fill = lightgray,label=above right:$\Ptwo$] { } edge from parent
                [dashed,draw,
                    preaction={draw,lightgray,-,
                    double=lightgray,
                    double distance=20\pgflinewidth,
                    }]
                child{ node { } edge from parent[dashed]
                child{ node[itri, xshift=1mm, yshift=3mm,dashed]{$\hat{T}_{1}$}
                edge from parent[dashed]
                }}
                child{
            node [fill=black] { } edge from parent
            [draw,solid,
                    preaction={
                    draw,solid,lightgray,-,
                    double=lightgray,
                    double distance=20\pgflinewidth,
                    }]
                child{ node { } edge from parent[solid]
                child{ node[itri, xshift=1mm, yshift=3mm,solid]{$T_{b}$} }    
                }
                child{ node { } edge from parent[solid]
                child{ node[itri, xshift=-1mm, yshift=2.875mm,solid]{$T_{d}$} }    
                }
                }
            }}}
            child { edge from parent[dashed] 
                node[itri, xshift=8mm, yshift=-9mm, dashed]{ } 
            }
        }
        child { edge from parent[dashed]
            node[itri, xshift=7mm, yshift=-8mm, dashed]{ } 
        };
    \end{tikzpicture}
    \\
    Case 2 \\
  \end{tabular}

  \caption{Top:  Case 1 of the proof of Theorem \ref{thm:minphif}. Here, $T_a$ is a sibling of $T_c$ in $T$, with size strictly larger than that of $T_b$, and $T_b$ is a sibling of a subtree of $T$ containing $T_a$. The highlighted path $\Path$ contains all vertices whose induced subtree sizes change when subtrees $T_b$ and $T_c$ are swapped. The dotted edges and subtrees may or may not exist in $T$. 
  Bottom:  Case 2 of the proof of Theorem \ref{thm:minphif}. Here, $T_a$ is a sibling of $T_c$ in $T$, with size strictly larger than that of $T_b$,  but not containing $T_b$. Similarly, $T_b$ is a sibling of $T_d$ in $T$, with size  strictly larger than that of $T_a$, but not  containing $T_a$. The highlighted paths $\Pone$ and $\Ptwo$ contain all vertices whose induced subtree sizes  change when $T_b$ and $T_c$ or $T_a$ and $T_d$ are swapped. The dotted edges and subtrees may or may not exist in $T$.}
  \label{fig:case1}
  \label{fig:case2}
\end{figure}

\begin{enumerate}
\item We start with Case 1 as depicted in Figure~\ref{fig:case1}. We construct a new tree $T'$ as follows: $T'$ is like $T$, but the subtrees $T_b$ and $T_c$ are interchanged. We let $\Path$ be the path highlighted in Figure~\ref{fig:case1}, where $\Path$ is the path containing all vertices $v$ for which we have $n_v(T)\neq n_v(T')$, where $n_v(T)$ and $n_v(T')$ denote the induced subtree sizes of $v$ in $T$ and $T'$, respectively. We note that for all vertices $v \in \mathring{V}(T)\setminus {\Path} $ we have $n_v(T)=n_v(T')$, as these vertices' subtree sizes are not affected by the subtree swap of $T_b$ and $T_c$. This reasoning leads to the following observation:

\begin{align}\label{eq:phifcase1}
\Phi_f(T')
&= \Phi_f(T)-\sum\limits_{v \in \Path} \left(f(n_v(T))- f(n_v(T)-n_c+n_b)\right),
\end{align}
where $n_b$ and $n_c$ are the subtree sizes of $T_b$ and $T_c$, respectively. 

Now, we know that $n_c>n_b$, because if $n_c< n_b$, the GFB algorithm would not have merged $T_a$ and $T_b$ (but $T_a$ and $T_c$ instead). Moreover, if we had $n_b=n_c$, then $T_b$ and $T_c$ would be isomorphic, as all steps prior to merging $T_a$ and $T_b$ worked in $T$ in the same way as in $T_n^{gfb}$ (by choice of $(T_a,T_b)$ as the minimal subtree of $T_n^{gfb}$ which could not be formed to build $T$).

In this case, however, $(T_a,T_b)$ would be isomorphic to $(T_a,T_c)$ and thus be contained in $T$, a contradiction to our choice of $T_a$ and $T_b$. Thus, $n_c>n_b$ and therefore $n_v(T)-n_c+n_b<n_v(T)$ for all $v \in \mathcal{P}$. As $f$ is strictly increasing, this implies that $f(n_v(T))> f(n_v(T)-n_c+n_b)$. With Equation \eqref{eq:phifcase1}, this directly leads to:

\begin{align}\
\Phi_f(T')
&= \Phi_f(T)-\sum\limits_{v \in \Path} \underbrace{\left(f(n_v(T))- f(n_v(T)-n_c+n_b)\right)}_{>0}<\Phi_f(T).
\end{align}

Thus, we have $\Phi_f(T')<\Phi_f(T)$, which contradicts the minimality of $T$ and thus completes this case.

\item Next, we let $T$ be as depicted as Case 2 in Figure \ref{fig:case2}, with $T_a$ a sibling to some tree $T_c$ not containing $T_b$, and $T_b$  a sibling to some tree $T_d$ not containing $T_a$. We now construct two trees $T'$ and $T''$ as follows: $T'$ is like $T$, but subtrees $T_b$ and $T_c$ are swapped. Similarly, $T''$ is like $T$, but subtrees $T_a$ and $T_d$ are swapped. We denote by $n_a$, $n_b$, $n_c$ and $n_d$ the number of leaves in $T_a$, $T_b$, $T_c$ and $T_d$, respectively. As in the first case, we necessarily have $n_d>n_a$ and  $n_c>n_b$ by choice of $T_a$ and $T_b$. 

We let $T_1,\ldots, T_k$ be the subtrees pending on path $\Pone$ as depicted in Figure \ref{fig:case2} if these subtrees exist in $T$. Similarly, we let $\widehat{T}_1,\ldots, \widehat{T}_l$ be the subtrees pending on path $\Ptwo$ as depicted in Figure \ref{fig:case2} if these subtrees exist in $T$. We denote by $n_i$ the number of leaves in $T_i$ for $i=1,,\ldots,k$, and by $\widehat{n}_i$ the number of leaves in $\widehat{T}_i$ for $i=1,\ldots,l$. Moreover, we define $n_0=\widehat{n}_0=0$. With these definitions, we can now introduce variables $t_i$ and $\widehat{t}_i$ defined as follows: $t_i=\sum\limits_{j=0}^i n_j$ for $i=0,\ldots, k$ and $\widehat{t}_i=\sum\limits_{j=0}^i \widehat{n}_j$ for $i=0,\ldots, l$. 

\begin{itemize}
\item We enumerate the vertices of $\Pone$ such that $v_0$ is the parent of $T_a$ and $T_c$ in $T$ and such that $v_i$ is the parent of $v_{i-1}$ for each $i=1,\ldots,k$ if $k>0$; that is, if trees $T_1,\ldots,T_k$ exist in $T$. Then, for the subtree sizes $n_{v_i}$ we derive: \begin{itemize} \item $n_{v_i}(T) = n_a+n_c+t_i$ for $i=0,\ldots,k$, 
\item $n_{v_i}(T') = n_a+n_b+t_i$ for $i=0,\ldots,k$,
\item $n_{v_i}(T'') = n_c+n_d+t_i$ for $i=0,\ldots,k$.
\end{itemize}

\item Similarly, we enumerate the vertices of $\Ptwo$ such that $w_0$ is the parent of $T_b$ and $T_d$ in $T$ and such that $w_i$ is the parent of $w_{i-1}$ for each $i=1,\ldots,l$ if $l>0$; that is, if trees $\widehat{T}_1,\ldots,\widehat{T}_l$ exist in $T$. Then, for the subtree sizes $n_{w_i}$ we derive: \begin{itemize} \item $n_{w_i}(T) = n_b+n_d+\widehat{t}_i$ for $i=0,\ldots,l$, 
\item $n_{w_i}(T') = n_c+n_d+\widehat{t}_i$ for $i=0,\ldots,l$,
\item $n_{w_i}(T'') = n_a+n_b+\widehat{t}_i$ for $i=0,\ldots,l$.
\end{itemize}
\end{itemize}

\par\vspace{0.2cm}
We now set $\lambda=\frac{n_d-n_a}{n_d-n_a+n_c-n_b}$. Since we have $n_d>n_a$ and $n_c>n_b$, we have  $\lambda \in (0,1)$. We now show that this choice of $\lambda$ has two additional properties, which will be useful regarding the concavity of $f$:

\begin{itemize}
\item We have $n_b+n_d+\widehat{t}_i = \lambda(n_c+n_d+\widehat{t}_i)+(1-\lambda)(n_a+n_b+\widehat{t}_i)$ for all $i=0,\ldots,l$:

\begin{align*}
&\lambda(n_c+n_d+\widehat{t}_i)+(1-\lambda)(n_a+n_b+\widehat{t}_i)\\&=\lambda n_c+\lambda n_d -\lambda n_a - \lambda n_b + n_a + n_b + \widehat{t}_i \\
&= n_a+n_b+\widehat{t}_i+\lambda(n_d-n_a+n_c-n_b)\\
&= n_a+n_b+\widehat{t}_i+\frac{n_d-n_a}{n_d-n_a+n_c-n_b} (n_d-n_a+n_c-n_b)\\
&=n_a+n_b+\widehat{t}_i+n_d-n_a\\
&=n_b+n_d+\widehat{t}_i.
\end{align*}

\item Analogously, we have $n_a+n_c+t_i = \lambda(n_a+n_b+t_i)+(1-\lambda)(n_c+n_d+t_i)$ for all $i=0,\ldots,k$.
\end{itemize}

\par\vspace{0.2cm}

The first one of these two  points above shows that there exists  $\lambda \in [0,1]$ with $n_b+n_d+\widehat{t}_i = \lambda(n_c+n_d+\widehat{t}_i)+(1-\lambda)(n_a+n_b+\widehat{t}_i)$, so that we have for all $i=0,\ldots, l$:  \begin{align}\label{eq:concave1}
f(n_b+n_d+\widehat{t}_i)&= f(\lambda(n_c+n_d+\widehat{t}_i)+(1-\lambda)(n_a+n_b+\widehat{t}_i))\\ &>\lambda f(n_c+n_d+\widehat{t}_i) +(1-\lambda)f(n_a+n_b+\widehat{t}_i)\nonumber,\end{align}
where the inequality holds due to the strict concavity of $f$. Analogously, by the second point, we have for all $i=0,\ldots, k$: 

\begin{align}\label{eq:concave2}
f(n_a+n_c+t_i)&= f(\lambda(n_a+n_b+t_i)+(1-\lambda)(n_c+n_d+t_i))\\ &>\lambda f(n_a+n_b+t_i) +(1-\lambda)f(n_c+n_d+t_i).\nonumber\end{align}

Now we are finally in a position to derive a contradiction, namely by investigating the term $\Phi_f(T)-\lambda \Phi_f(T')-(1-\lambda)\Phi_f(T'')$ in two different ways.

\begin{enumerate}
\item By assumption, $T$ is a minimizer of $\Phi_f$, so we  have that $\Phi_f(T')\geq \Phi_f(T)$ as well as $\Phi_f(T'')\geq \Phi_f(T)$. Thus: 

\begin{align}\label{eq:easycase}\Phi_f(T)-\lambda \Phi_f(T')-(1-\lambda)\Phi_f(T'') & \leq \Phi_f(T)-\lambda \Phi_f(T)-(1-\lambda)\Phi_f(T)=0.
\end{align}

\item We now split the sum of $\Phi_f(T)=\sum\limits_{v \in \mathring{V}(T)}f(n_v)$ into three partial sums, namely the inner vertices belonging to $\mathcal{P}_1$, the ones belonging to $\mathcal{P}_2$ and the ones belonging to neither one of the paths. Note that as all vertices that are not contained in any one of the paths are not affected by the swaps leading from $T$ to $T'$ or $T''$, respectively, the last sum is the same for $\Phi_f(T)$, $\Phi_f(T')$ and $\Phi_f(T'')$. From our above observations concerning the subtree sizes $n_{v_i}$ of  $\mathcal{P}_1$ and $n_{w_i}$ of $\mathcal{P}_2$, we derive:

\begin{align}
&\Phi_f(T)-\lambda \Phi_f(T')-(1-\lambda)\Phi_f(T'') \\&=
\sum\limits_{v \in \Pone} f(n_v(T)) + \sum\limits_{v \in \Ptwo} f(n_v(T)) +\sum\limits_{v \in \mathring{V}\setminus \{\Pone,\Ptwo\}} f(n_v(T)) \nonumber \\
&-\lambda \cdot \sum\limits_{v \in \Pone} f(n_v(T')) -\lambda \cdot  \sum\limits_{v \in \Ptwo} f(n_v(T')) \nonumber \\ &-\lambda \cdot \sum\limits_{v \in \mathring{V}\setminus \{\Pone,\Ptwo\}} f(n_v(T'))  \nonumber -(1-\lambda) \cdot \sum\limits_{v \in \Pone} f(n_v(T'')) \nonumber \\ 
& -(1-\lambda) \cdot \sum\limits_{v \in \Ptwo} f(n_v(T'')) \nonumber \\ &-(1-\lambda)\cdot \sum\limits_{v \in \mathring{V}\setminus{\{\Pone,\Ptwo\}}} f(n_v(T'')) \nonumber 
\end{align}

\begin{align}
&=\sum\limits_{i=0}^k \underbrace{\left( f(n_a+n_c+t_i)-\lambda f(n_a+n_b+t_i) -(1-\lambda) f(n_c+n_d+t_i) \right)}_{>0 \mbox{ by Eq. \eqref{eq:concave2}}}   \nonumber \\
&+\sum\limits_{i=0}^l \underbrace{\left( f(n_b+n_d+\widehat{t}_i)-\lambda f(n_c+n_d+\widehat{t}_i) -(1-\lambda) f(n_a+n_b+\widehat{t}_i) \right)}_{>0 \mbox{ by Eq. \eqref{eq:concave1}}} &  \nonumber \\
&>0 \label{eq:contradiction}.
\end{align}

The obvious contradiction between Inequalities \eqref{eq:easycase}, which states that $\Phi_f(T)-\lambda \Phi_f(T')-(1-\lambda)\Phi_f(T'') \leq 0$ and \eqref{eq:contradiction}, which states that $\Phi_f(T)-\lambda \Phi_f(T')-(1-\lambda)\Phi_f(T'') > 0$, shows that our assumption concerning the existence of $T$ must have been wrong. In fact, this contradiction shows that at least one of the two trees $T'$ and $T''$ must have a lower $\Phi_f$ value than $T$. This completes the proof and thus shows that $T_n^{gfb}$ is the unique tree minimizing $\Phi_f$. 
\end{enumerate}

\end{enumerate}

\end{proof}

Before we investigate the implications of Theorem \ref{thm:minphif} on imbalance indices, we derive the following corollary.

\begin{corollary}\label{cor:prodmin} Let $n\in \mathbb{N}$, $n \geq 2$ and let $c \in \mathbb{R}$ with $c > -2$. Let $f:\mathbb{R}_{\geq 2} \rightarrow \mathbb{R}$ be a strictly increasing function. Then, we have that $T_n^{gfb}$ is the unique minimizer of $f(\pi_c(T))$ among all rooted binary trees $T$ with $n$ leaves, where 
$\pi_c(T)= \prod\limits_{v \in \mathring{V}(T)} (n_v+c)$. In particular, this holds for the identity function, $f(x)=x$ for all $x \in \mathbb{R}_{\geq 2}$.
\end{corollary}

\begin{proof} We start with considering the product function $\pi_c$. We let $c$ and $n$ be as described in the corollary. We have  $n_v\geq 2$ for all inner nodes $v$ of a rooted binary tree $T$, as the smallest possible subtree size is 2 (which is the case in which $v$ is the parent of a cherry). Thus, we have $n_v+c>0$ for all $v \in \mathring{V}$, as $c > -2$ by assumption. This, however, means that all factors in $\pi_c(T)$ are strictly larger than 0, which shows that $\pi_c(T)>0$. This, in turn, means that $\log(\pi_c(T))$ is defined. 

Now we consider this term further:  $$ \log(\pi_c(T))=\log\left( \prod\limits_{v \in \mathring{V}(T)} (n_v+c) \right) = \sum\limits_{v \in \mathring{V}(T)} \log(n_v+c).$$
As the logarithm is strictly concave and strictly monotonically increasing, we know by Theorem \ref{thm:minphif} that the latter sum is uniquely minimized by $T_n^{gfb}$. Thus, the same applies to $\log(\pi_c(T))$. However, by the strict monotonicity of $\log$, the minimum of $\log(\pi_c(T))$ is reached precisely when the minimum of $\pi_c(T)$ is reached, which shows that $T_n^{gfb}$ is also the unique tree minimizing $\pi_c(T)$. 

Now, for any strictly increasing function $f:\mathbb{R}_{\geq 2}\rightarrow \mathbb{R}$, we have that $T_n^{gfb}$ is also the unique minimizer of $f(\pi_c(T))$ due to the monotonicity of $f$. This completes the proof. 
\end{proof}

\subsection{Implications of the extremal GFB properties on measures of tree balance} \label{sec:bal} The main aim of this section is two-fold: First, we want to show that both functions $\Phi_f$ and $\pi_c$ as defined in the previous section form families of imbalance indices for certain choices of $f$ and $c$, respectively. We will continue to show that the imbalance index family based on the product function and a constant $c$ is merely a subfamily of the imbalance index family based on strictly increasing and strictly concave functions $f$. 

Then, we want to use our  findings to characterize all trees minimizing the $\widehat{s}$-shape and $Q$-shape statistics, thus answering several open questions from \cite{fischerBook}. 

However, in order to show that a function is an imbalance index, analyzing the minimum as in the previous section is not sufficient. Instead, we also need to investigate the caterpillar in order to investigate the maximum. We start with $\Phi_f$. Note that the following theorem can already be found in \cite[Theorem 4.7]{HamannThesis}, albeit with a different proof.

\begin{theorem}\label{thm:maxphif} Let $n\in \mathbb{N}$, $ n \geq 2$ and let $f:\mathbb{R}_{\geq 2}\rightarrow \mathbb{R}$ be strictly monotonically increasing, with $f(n_1)>f(n_2)$ if and only if $n_1>n_2$. We consider $\Phi_f(T)= \sum\limits_{v \in \mathring{V}(T)} f(n_v)$. Then, $T_n^{cat}$ is the unique tree maximizing $\Phi_f$.  
\end{theorem}

\begin{proof} 
We prove the statement by contradiction. We suppose that there is a non-caterpillar tree $T$ maximizing $\Phi_f$. We choose the smallest possible $n$ for which such a tree $T$ with $n$ leaves exists.  Thus, for all numbers smaller than $n$, the unique maximizer of $\Phi_f$ is the caterpillar. In particular, this shows that $n\geq 4$, because for any value smaller than 4, there is only one rooted binary tree.

We let $T=(T_1,T_2)$ be the standard decomposition of $T$. Then, $\Phi_f(T)=\Phi_f(T_1)+\Phi_f(T_2)+ f(n)$, where the last summand $f(n)$ results from the root $\rho$ of $T$. This equality shows that $T$ can only  maximize $\Phi_f$ among all trees with $n$ leaves if $T_1$ and $T_2$ maximize $\Phi_f$ among all trees with $n_1$ and $n_2$ leaves, respectively, where $n_1$ is the number of leaves of $T_1$ and $n_2$ is the number of leaves of $T_2$. Thus, as we chose $T$ to be a counterexample of minimal size concerning the statement of the theorem, we know by assumption that $T_1$ and $T_2$ must be caterpillars. Note that this implies that $n_1\geq 2$ and $n_2 \geq 2$ (which is possible as $n_1+n_2=n\geq 4$), because if we had $n_2=1$ and $T_1$ is a caterpillar, then $T$ would be a caterpillar also. The same would happen if $n_1=1$.

Thus, we know that, as $T_1$ and $T_2$ are caterpillars with at least two leaves each, each of them has precisely one cherry. Let $[a_1,b_1]$ denote the cherry of $T_1$ and $[a_2,b_2]$ denote the cherry of $T_2$. The parents of $[a_1,b_1]$ and $[a_2,b_2]$ are denoted by $v_0$ and $w_0$, respectively. Note that on the path from $v_0$ to the root $\rho$ of $T$, there might be more vertices $v_1,\ldots,v_k$, all of which -- if they exist -- are adjacent to a leaf as $T_1$ is a caterpillar. Analogously, there might be more vertices $w_1,\ldots,w_l$ on the path from $w_0$ to $\rho$, all of which -- if they exist -- are adjacent to a leaf as $T_2$ is a caterpillar. Note that this means that $T$ looks as depicted on the left-hand side of Figure \ref{fig:catmax}. We denote the leaves adjacent to $v_i$ with $x_i$ for $i=1,\ldots,k$ (if they exist), and the leaves adjacent to $w_i$ with $y_i$ for $i=1,\ldots, l$ (if they exist).

\begin{figure}
  \centering
  \begin{tabular}{ccc}
    \begin{tikzpicture}[
      level distance=.75cm,
      sibling distance=1.5cm,
      level/.style={sibling distance={2.5cm/max(1,#1)}}
    ]
    \node {$\rho$} 
      child { node {$v_k$} edge from parent
          child { node {$v_1$} edge from parent [dotted] 
              child { node {$v_0$} edge from parent [solid] 
                  child { node {$a_1$} edge from parent }
                  child { node {$b_1$} edge from parent }
              }
              child { node {$x_1$} edge from parent [dotted] }
          }
          child { node {$x_k$} edge from parent [dotted] }
      }
      child { node[draw=black,rectangle,thick,inner sep=2pt] {$\mathbf{w_l}$} edge from parent
          child { node{$y_l$} edge from parent [dotted]}
          child { node{$w_{l-1}$} edge from parent [solid]
              child { node{$y_{l-1}$} edge from parent [dotted]}
              child { node {$w_1$} edge from parent [dotted]
                  child { node{$y_{1}$} edge from parent [dotted]}
                  child { node {$w_0$} edge from parent [solid]
                      child { node {$a_2$} edge from parent }
                      child { node {$b_2$} edge from parent }
                  }
              }        
          }
      };
    \end{tikzpicture} 
    &&
    \begin{tikzpicture}[
      level distance=.75cm,
      sibling distance=1.5cm,
      level/.style={sibling distance={2.5cm/max(1,#1)}}
    ]
    \node {$\rho$} 
      child { node[draw=black,rectangle,thick,inner sep=2pt] {$\mathbf{v_{k+1}}$} edge from parent
          child { node {$v_k$} edge from parent
              child { node {$v_1$} edge from parent [dotted] 
                  child { node {$v_0$} edge from parent [solid] 
                      child { node {$a_1$} edge from parent }
                      child { node {$b_1$} edge from parent }
                  }
                  child { node {$x_1$} edge from parent [dotted] }
              }
              child { node {$x_k$} edge from parent [dotted] }
          }
          child {node{$y_l$} edge from parent [dotted] }
      }
      child { node{$w_{l-1}$} edge from parent [solid]
          child { node{$y_{l-1}$} edge from parent [dotted]}
          child { node {$w_1$} edge from parent [dotted]
              child { node{$y_{1}$} edge from parent [dotted]}
              child { node {$w_0$} edge from parent [solid]
                  child { node {$a_2$} edge from parent }
                  child { node {$b_2$} edge from parent }
              }
          }        
      };
    \end{tikzpicture}   
    \\
    $T$ && $T'$\\
  \end{tabular}

  \caption{Trees $T$ and $T'$ as described in the proof of Theorem \ref{thm:maxphif}. The only differences between the subtree sizes are at nodes $w_l$ and $v_{k+1}$, highlighted with a box.}
  \label{fig:catmax}
\end{figure}
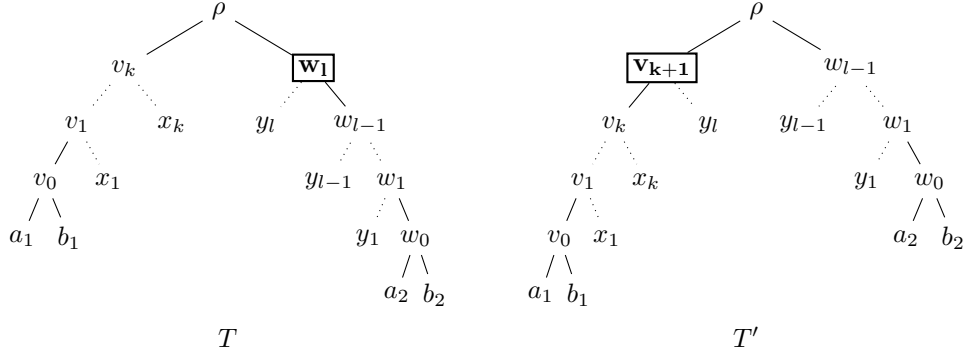

Now we assume without loss of generality that $k\geq l$ and consider $w_l$ (note that $l$ might be 0) and consider a leaf $z_l$ adjacent to $w_l$ (note that $z_l$ might be either $a_2$ or $b_2$ if $l=0$; otherwise we have $z_l=y_l$). We now create a tree $T'$ by deleting edge $(w_l,z_l)$, subdividing edge $(\rho,v_k)$ by introducing a new degree-2 vertex $v_{k+1}$ and then adding a new edge $(v_{k+1},z_l)$ and suppressing $w_l$. The resulting tree $T'$ is  depicted on the right-hand side of Figure \ref{fig:catmax}. Note that between $T$ and $T'$, all subtree sizes are equal except for that of $v_{k+1}$, which equals $k+3$ in $T'$ as above but does not exist in $T$, and that of $w_l$, which equals $l+2$ (stemming from the fact that $w_l$ is ancestral to $a_2$, $b_2$ and $y_1,\ldots,y_l$ if the latter leaves exist). Thus, we derive for $\Phi_f(T')$:
$$ \Phi_f(T') =  \Phi_f(T) + f(k+3)-f(l+2),$$
which immediately shows that $\Phi_f(T')>\Phi_f(T)$, as we assumed that $k\geq l$. This contradicts the maximality of $T$ and thus completes the proof. 
\end{proof}

The maximality of the caterpillar now immediately leads to the following corollary, which shows that each $\Phi_f$  is indeed an imbalance index.

\begin{corollary}\label{cor:phifIndex} Let $n \in \mathbb{N}, n \geq 2$ and let  $f:\mathbb{R}_{\geq 2}\rightarrow \mathbb{R}$ be a strictly increasing and strictly concave function. Then,  $\Phi_f$ is an imbalance index. Moreover, the GFB tree $T_n^{gfb}$ is the only minimizer of this function in $\BTnstar$.
\end{corollary}

\begin{proof} This is a direct consequence of the definition of an imbalance index in combination with Theorems \ref{thm:minphif} and \ref{thm:maxphif} and the fact that $T_h^{fb}=T_n^{gfb}$ if $n=2^h$.
\end{proof}

We now use Corollary~\ref{cor:phifIndex} to answer several open problems from the literature, most notably from \cite{fischerBook}. Note that while it is already known that the $\widehat{s}$-shape statistic is an imbalance index \cite{fischerBook}, this has not been formally proven yet for the $Q$-shape statistic. However, while for the $\widehat{s}$-shape statistic the minima have already been known for the cases in which $n=2^h$ (it is $T_n^{fb}$ as $\widehat{s}$ is an imbalance index), it has not been known what the minimal trees are in cases in which $n$ is not a power of two \cite{fischerBook}. The following corollary fully characterizes these minima both for $\widehat{s}$ and $Q$ by showing that in both cases, $T_n^{gfb}$ is the unique minimizer.

\begin{corollary}\label{cor:sshape} 
The $\widehat{s}$-shape statistic and the $Q$-shape statistic are both tree imbalance indices with the property that the GFB tree $T_n^{gfb}$ is their only minimizer in $\BTnstar$ for any value of $n \in \mathbb{N}, n \geq 2$.
\end{corollary}

\begin{proof} Let $n \in \mathbb{N}, n \geq 2$. We define $f_{\widehat{s}}(i)=\log(i-1)$ and $f_Q(i)=\log(i)$ for 
$i\in \mathbb{R}_{\geq 2}$. Note that $f_{\widehat{s}}$ and $f_Q$ are both strictly increasing and strictly concave. Now, by definition of $\widehat{s}$ and $Q$, we have for all rooted binary trees $T$ with $n$ leaves: 

\begin{align*}
\Phi_{f_{\widehat{s}}(T)} &= \sum\limits_{v \in \mathring{V}(T)} f_{\widehat{s}}(n_v)=\sum\limits_{v \in \mathring{V}(T)}\log(n_v-1)=\widehat{s}(T), 
\end{align*}

as well as

\begin{align*}
\Phi_{f_Q(T)} &= \sum\limits_{v \in \mathring{V}(T)} f_{Q}(n_v)=\sum\limits_{v \in \mathring{V}(T)}\log(n_v)=Q(T). 
\end{align*}

Applying Corollary \ref{cor:phifIndex} completes the proof.  
\end{proof}

So, we now know that for all values of $n$, there is only one tree minimizing the $\widehat{s}$-shape statistic (thus answering the question concerning the number of minima posed in \cite[Chapter 9]{fischerBook}), and we have fully characterized this unique minimum as $T_n^{gfb}$. In Section \ref{sec:explicitformulas}, we will also deliver explicit formulas to calculate the minimal value of $\widehat{s}$ for all $n$. 

We now turn our attention to the product function to show that functions of this type also form a family of tree imbalance indices. We again start with considering the caterpillar.

\begin{corollary}\label{cor:prodmaxcat} 
Let $n\in \mathbb{N}, n \geq 2$ and let $c \in \mathbb{R}$, $c> -2$. Then, we have that $T_n^{cat}$ is the unique maximizer of $\pi_c(T)$ in $\BTnstar$, where 
$\pi_c(T)= \prod\limits_{v \in \mathring{V}(T)} (n_v+c)$. 
\end{corollary}

\begin{proof} Let $c \in \mathbb{R}$, $c> -2$. 
We can set $f_c(i)=\log(i+c)$ for $i\in \mathbb{R}_{\geq 2}$. Then, $f_c$ is both strictly concave and strictly increasing (as the logarithm with base 2 has these properties). With this function $f_c$, we have for any rooted binary tree $T$ with $n$ leaves: 

\begin{align}\label{eq:phifc}\Phi_{f_c}(T)=\sum\limits_{v \in \mathring{V}(T)}f_c(n_v)=\sum\limits_{v \in \mathring{V}(T)}\log(n_v+c)=\log\left(\prod\limits_{v \in \mathring{V}(T)}(n_v+c)\right)=\log(\pi_c(T)).\end{align} 

By Theorem \ref{thm:maxphif} we know that $T_n^{cat}$ is the unique maximizer of $\Phi_{f_c}$, and, thus, we can conclude  that $\log(\pi_c(T_n^{cat}))>\log(\pi_c(T))$ for all rooted binary trees $T$ with $n$ leaves. By the monotonicity of $\log$, this directly implies $\pi_c(T_n^{cat})>\pi_c(T)$ for all such trees $T$. This concludes the proof.
\end{proof}

\par\vspace{0.2cm}

We now use Corollary~\ref{cor:prodmaxcat} to show that the product function leads to a family of imbalance indices.

\begin{corollary}\label{cor:ClassProd2} Let $n \in \mathbb{N}$, $n\geq 2$ and let $c \in \mathbb{R}, c> -2$. Moreover, let $f:\mathbb{R}_{\geq 2}\rightarrow \mathbb{R}$ be a strictly increasing function. Then,  $f(\pi_c)$ is an imbalance index. Moreover, the GFB tree $T_n^{gfb}$ is the only minimizer of this function.
\end{corollary}

\begin{proof} The fact that $T_n^{cat}$ is the unique maximizer of $\pi_c$ follows from  Corollary \ref{cor:prodmaxcat}, which in turn shows by the monotonicity of $f$ that $T_n^{cat}$ is the unique maximizer of $f(\pi_c)$. The fact that $T_n^{gfb}$ is the unique minimizer of $f(\pi_c)$ was shown in Corollary \ref{cor:prodmin}. Using the fact that $T_h^{fb}=T_n^{gfb}$ if $n=2^h$ and the definition of an imbalance index thus completes the proof.
\end{proof}

Corollary \ref{cor:ClassProd2} shows that the product functions are a family of tree imbalance indices.  We further classify this family as merely a subfamily of the family of tree imbalance indices $\Phi_f$ in the sense that the tree rankings from balanced to imbalanced induced by these indices coincide with rankings induced by members of the $\Phi_f$ family, as the following proposition shows.

\begin{proposition}\label{prop:subfam} Let $c\in \mathbb{R}, c> -2$ and $n\in \mathbb{N}$, $n\geq 2$, let $f$ be a strictly increasing function. We consider the imbalance index $f(\pi_c)$ and its induced ranking of trees $T_n^{gfb},\ldots,T_n^{cat}$ from balanced to imbalanced. Then, there is a strictly increasing and strictly concave function $f_c:\mathbb{R}_{\geq 2}\rightarrow \mathbb{R}$ such that $\Phi_{f_c}$ induces the same ranking as $f(\pi_c)$.
\end{proposition}

\begin{proof}
First we note that by the monotonicity of $f$, $f(\pi_c)$ induces the exact same ranking as $\pi_c$, which in turn induces the exact same ranking as $\log(\pi_c)$ by the monotonicity of the logarithm. 

Now, we set  $f_c(i)=\log(i+c)$ for  $i\in \mathbb{R}_{\geq 2}$. Then, as in the proof of Corollary \ref{cor:prodmaxcat}, $f_c$ is both strictly concave and strictly increasing, and just as in Equation \eqref{eq:phifc} we can conclude $\Phi_{f_c}(T)=\log(\pi_c(T))$ for all rooted binary trees $T$ with $n$ leaves. This shows that $\Phi_{f_c}$ and $f(\pi_c)$ induce the same rankings of all rooted binary trees with $n$ leaves and thus completes the proof.
\end{proof}

Before we turn our attention to deriving explicit formulas for the minimal values of $\widehat{s}$ and $Q$ in Section \ref{sec:explicitformulas}, we investigate the family of tree imbalance indices $\pi_c$ further. We first show that the family contains infinitely many different members in the sense that for choices of real $c_1,c_2$ larger than $-2$ with $c_1\neq c_2$, we can find two trees $T$ and $T'$ such that $\pi_{c_1}$ will consider $T$ as more imbalanced than $T'$ and $\pi_{c_2}$ gives the opposite ranking. Thus, there is an uncountably infinite family of genuinely distinct imbalance indices.

\begin{proposition}\label{prop:inf} Let $c_1, c_2 \in \mathbb{R}$ be distinct with $c_1,c_2>-2$. Then there exist two trees $T$ and $T'$ such that $\pi_{c_1}(T)>\pi_{c_1}(T')$ and $\pi_{c_2}(T)<\pi_{c_2}(T')$; that is, the imbalance indices $\pi_{c_1}$ and $\pi_{c_2}$ rank $T$ and $T'$ differently.
\end{proposition}

\begin{proof} We will construct $T$ and $T'$ as depicted in Figure \ref{fig:inf} with suitable choices of sizes of subtrees: $n_{11}$ of $T_{11}$, $n_{12}$ of $T_{12}$, $n_{21}$ of $T_{21}$ and $n_{22}$ of $T_{22}$, respectively. Subtrees $T_{11}$, $T_{12}$, $T_{21}$ and $T_{22}$ can then be chosen arbitrarily as long as they have the respective numbers of leaves.

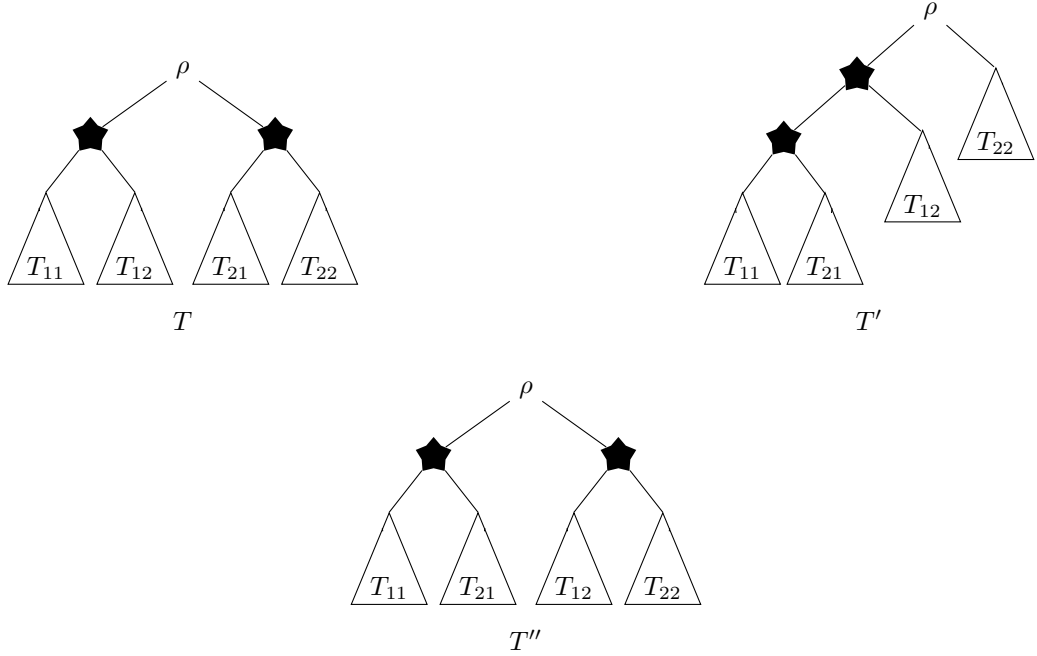
\begin{figure}[ht]
  \centering

  \begin{minipage}[t]{0.45\textwidth}
    \centering
    \begin{tikzpicture}[
      level distance=.875cm,
      level/.style={sibling distance={2.75cm/max(1.125,#1)}}
    ]
      \node {$\rho$}
        child{ node [star, fill = black]{}
          edge from parent
          child{ node {} child{ node[itri,xshift=1mm]{$T_{11}$} } }
          child{ node {} child{ node[itri,xshift=-1mm]{$T_{12}$} } }
        }
        child{ node [star, fill = black]{}
          child{ node {} child{ node[itri,xshift=1mm]{$T_{21}$} } }
          child{ node {} child{ node[itri,xshift=-1mm]{$T_{22}$} } }
        };
    \end{tikzpicture}

    \vspace{0.5em}
    $T$
  \end{minipage}
  \hfill
  \begin{minipage}[t]{0.45\textwidth}
    \centering
    \begin{tikzpicture}[
      level distance=.85cm,
      level/.style={sibling distance={3.875cm/max(2,#1)}}
    ]
      \node {$\rho$}
        child{ node [star, fill = black]{} edge from parent
          child{ node [star, fill = black]{}
            child{ node {} child{ node[itri, xshift=1mm]{$T_{11}$} } }
            child{ node {} child{ node[itri, xshift=-1mm]{$T_{21}$} } }
          }                
          child{ node {} child{ node[itri, xshift=-1mm, yshift=-.25mm]{$T_{12}$} } }
        }
        child{ node {} child{ node[itri, xshift=-1mm, yshift=-.5mm]{$T_{22}$} } };
    \end{tikzpicture}

    \vspace{0.5em}
    $T'$
  \end{minipage}

  \vspace{1.5em}

  \begin{minipage}[t]{0.6\textwidth}
    \centering
    \begin{tikzpicture}[
      level distance=.875cm,
      level/.style={sibling distance={2.75cm/max(1.125,#1)}}
    ]
      \node {$\rho$}
        child{ node [star, fill = black]{} edge from parent
          child{ node {} child{ node[itri, xshift=1mm]{$T_{11}$} } }
          child{ node {} child{ node[itri, xshift=-1mm]{$T_{21}$} } }
        }
        child{ node [star, fill = black]{}
          child{ node {} child{ node[itri,xshift=1mm]{$T_{12}$} } }
          child{ node {} child{ node[itri,xshift=-1mm]{$T_{22}$} } }
        };
    \end{tikzpicture}

    \vspace{0.5em}
    $T''$
  \end{minipage}

  \caption{Trees $T$, $T'$, and $T''$ as needed in the proofs of Propositions \ref{prop:inf} and \ref{prop:equal}. Note that all three trees share the same subtrees $T_{11}$, $T_{12}$, $T_{21}$, and $T_{22}$, which are depicted schematically as triangles. The stars depict the inner vertices that play an important role in the proofs.}
  \label{fig:inf}
\end{figure}

Since the subtree sizes induced by the two trees $T$ and $T'$ only differ in two nodes, we can easily express $\pi_c(T')$ using $\pi_c(T)$ for any  choice of $c \in \mathbb{R}$ as follows:

\begin{align*}
\pi_c(T')=\pi_c(T) \cdot \frac{(n_{11}+n_{21}+c)\cdot(n_{11}+n_{12}+n_{21}+c)}{(n_{11}+n_{12}+c)\cdot (n_{21}+n_{22}+c)}.
\end{align*}
This implies:  \begin{align}\label{eq:inf}&\hspace{1.2cm} \pi_c(T')\gtrless \pi_c(T) \\& \Longleftrightarrow (n_{11}+n_{21}+c)\cdot(n_{11}+n_{12}+n_{21}+c) \gtrless (n_{11}+n_{12}+c)\cdot (n_{21}+n_{22}+c) \nonumber \\ &\Longleftrightarrow cn_{11} + cn_{21}-cn_{22}\gtrless n_{11}n_{22}+n_{12}n_{22}-n_{11}^2-n_{11}n_{12}-n_{11}n_{21}-n_{21}^2 \nonumber\\
&\Longleftrightarrow c \gtrless \frac{n_{11}n_{22}+n_{12}n_{22}-n_{11}^2-n_{11}n_{12}-n_{11}n_{21}-n_{21}^2}{n_{11} + n_{21}-n_{22}},
\end{align} 
(where the $\gtrless$-symbol stands for either $>$ or $<$ consistently throughout). 

The proof strategy now is to show that for any choice of $c_1,c_2$ larger than $-2$ we can choose $n_{11}$, $n_{12}$, $n_{12}$, $n_{21}$ and $n_{22}$ such that the fraction of Equation \eqref{eq:inf} lies between $c_1$ and $c_2$. By Equation \eqref{eq:inf}, this will show that $\pi_{c_1}(T')<\pi_{c_1}(T)$ and $\pi_{c_2}(T')>\pi_{c_2}(T)$ and thus conclude the proof. In the following, we assume without loss of generality that $c_2>c_1$ by exchanging $T$ and $T'$ if needed.

So now we let $c_1,c_2\in \mathbb{R}$ with $c_1 >-2$ and $c_2>c_1$. We let $k \in \mathbb{N}$ be such that $k\cdot (c_2-c_1)>2$. This guarantees that the open interval $(k\cdot c_1,k\cdot c_2)$ contains two consecutive integers $m$ and $m+1$, where $m \in \mathbb{Z}$. This implies that $\frac{m}{k}$ and $\frac{m+1}{k}$ are contained in the open interval $(c_1,c_2)$. We have two rational numbers contained in $(c_1,c_2)$. We now consider their mean $\frac{2m+1}{2k}$. We have:

$$ c_1 < \frac{m}{k} < \frac{2m+1}{2k} < \frac{m+1}{k}< c_2.$$
Following our proof strategy, the proof is thus complete if we can show that we can choose $n_{11}$, $n_{12}$,  $n_{21}$ and $n_{22}$ such that the fraction of Equation \eqref{eq:inf} equals $\frac{2m+1}{2k}$.

We now set $n_{11}=1$, $n_{22}=2$, $n_{21}=2k+1$, $n_{12}=4k^2+2+6k+2m$. We first verify that these are all valid leaf numbers: that all of these numbers are natural. Clearly, this holds for $n_{11}$ and $n_{22}$. Moreover, recall that $k \in \mathbb{N}$, so $n_{21}$ is also natural. But as $m \in \mathbb{Z}$, $m<0$ could be possible. So we need to verify that $n_{12}$ is positive. However, we know that $\frac{m}{k}>c_1>-2$ by the choice of $c_1$ and $\frac{m}{k}$, respectively. This shows that $m>-2k$ and thus $2m>-4k$, which leads to $n_{12}=4k^2+2+6k+2m>4k^2+2+6k-4k=4k^2+2+2k \in \mathbb{N}$. So our choices of $n_{11}$, $n_{12}$, $n_{21}$ and $n_{22}$ result in four positive integers and can be realized as subtree sizes in trees. We show that with these choices, we indeed get that the fraction of Equation \eqref{eq:inf} equals $\frac{2m+1}{2k}$:

\begin{align*}&\frac{n_{11}n_{22}+n_{12}n_{22}-n_{11}^2-n_{11}n_{12}-n_{11}n_{21}-n_{21}^2}{n_{11} + n_{21}-n_{22}}\\ &=
\frac{2+2(4k^2+2+6k+2m)-1-(4k^2+2+6k+2m)-(2k+1)-(2k+1)^2}{1+(2k+1)-2} \\
&=\frac{1+(4k^2+2+6k+2m)-(2k+1)-(4k^2+4k+1)}{2k}\\
&=\frac{2m+1}{2k}.
\end{align*}
This completes the proof.
\end{proof}

Note that unsurprisingly, it is easier to make $\pi_{c_1}$ and $\pi_{c_2}$ disagree concerning the ranking of $T$ and $T'$ if $c_2-c_1$ is large. If the difference is larger than one, the value of $k$ chosen in the proof can be 1, which is as small as possible. As our choice of subtree sizes in the proof was $n_{11}=1$, $n_{22}=2$, $n_{21}=2k+1$, and  $n_{12}$ such that $n_{12}>4k^2+2+2k$, this shows that even if $k=1$, we already need $n=1+2+3+8=14$ leaves for our construction. It may be possible to have smaller examples showing different rankings, but it is clear that
no two indices $\pi_{c_1}$ and $\pi_{c_2}$ will rank all trees in the same order if $c_1\neq c_2$.

We note $T$ and $T'$ as used in the proof of Proposition \ref{prop:inf} only differ in two subtree sizes, yet they will be ranked differently by certain members of the $\pi_c$ family. On the other hand, there are always pairs of trees that are ranked identically for all choices of $c$. Most prominently, this is of course the case for $T_n^{cat}$ and $T_n^{gfb}$, but these differ in most subtree sizes. The next proposition shows that for all $n$ at least 6, there are pairs of $T$ and $T'$ such that these trees differ only in two subtree sizes and we have $\pi_c(T)>\pi_c(T')$ for all $c\in \mathbb{R}, c>-2$.

\begin{proposition}\label{prop:equal} Let $n\in \mathbb{N}$ with $n \geq 6$. Then there exist two trees $T$ and $T''$ which only differ in two subtree sizes such that $\pi_{c}(T)>\pi_{c}(T'')$ for all $c \in \mathbb{R}$, with $c>-2$.
\end{proposition}

\begin{proof} We give an explicit construction for two non-isomorphic trees $T$ and $T''$ which fulfill the condition. We begin by choosing $n_{ij}\in \mathbb{N}$ such that $n_{11}< n_{22}$ and $n_{12}< n_{21}$ and such that $n_{11}+n_{12}+n_{21}+n_{22}=n\geq 6$. These values of  $n_{ij}$ (with $i,j \in \{1,2\}$) will be used as leaf numbers for the subtrees $T_{ij}$ for $T$ and $T''$ as depicted in Figure \ref{fig:inf}. Note that the conditions on $n_{ij}$ guarantee that $T$ and $T''$ are not isomorphic. 

Moreover, note that as the subtree sizes induced by the two trees $T$ and $T''$ from Figure \ref{fig:inf} only differ in two nodes, we can easily express $\pi_c(T'')$ using $\pi_c(T)$ for any  choice of $c \in \mathbb{R}$:

\begin{align*}
\pi_c(T'')=\pi_c(T) \cdot \frac{(n_{11}+n_{21}+c)\cdot(n_{12}+n_{22}+c)}{(n_{11}+n_{12}+c)\cdot (n_{21}+n_{22}+c)}.
\end{align*}

This shows that we have:  \begin{align*}\pi_c(T'')\gtrless \pi_c(T) & \Longleftrightarrow (n_{11}+n_{21}+c)\cdot(n_{12}+n_{22}+c) \gtrless (n_{11}+n_{12}+c)\cdot (n_{21}+n_{22}+c)  \\
& \Longleftrightarrow n_{11}n_{12}+n_{21}n_{22} \gtrless n_{11}n_{21}+n_{12}n_{22} 
\end{align*} 
(where again the  $\gtrless$-symbol is a consistent inequality throughout). As the latter term is completely independent of $c$, this shows that all indices in the $\pi_c$ family will agree on how to rank $T$ and $T''$, and this will only be determined by these trees' subtree sizes. This completes the proof.
\end{proof}

Figure \ref{fig:n10cvariable} illustrates the differences of rankings within the $\pi_c$ family by considering the range of values $c=-1.99, -1.5, -1, -0.5, 0, 0.5, 1, 1.5, 2$ and comparing the rankings induced by these values for $n=10$. 

\begin{figure}[ht]
  \centering
  \includegraphics[scale=.3]{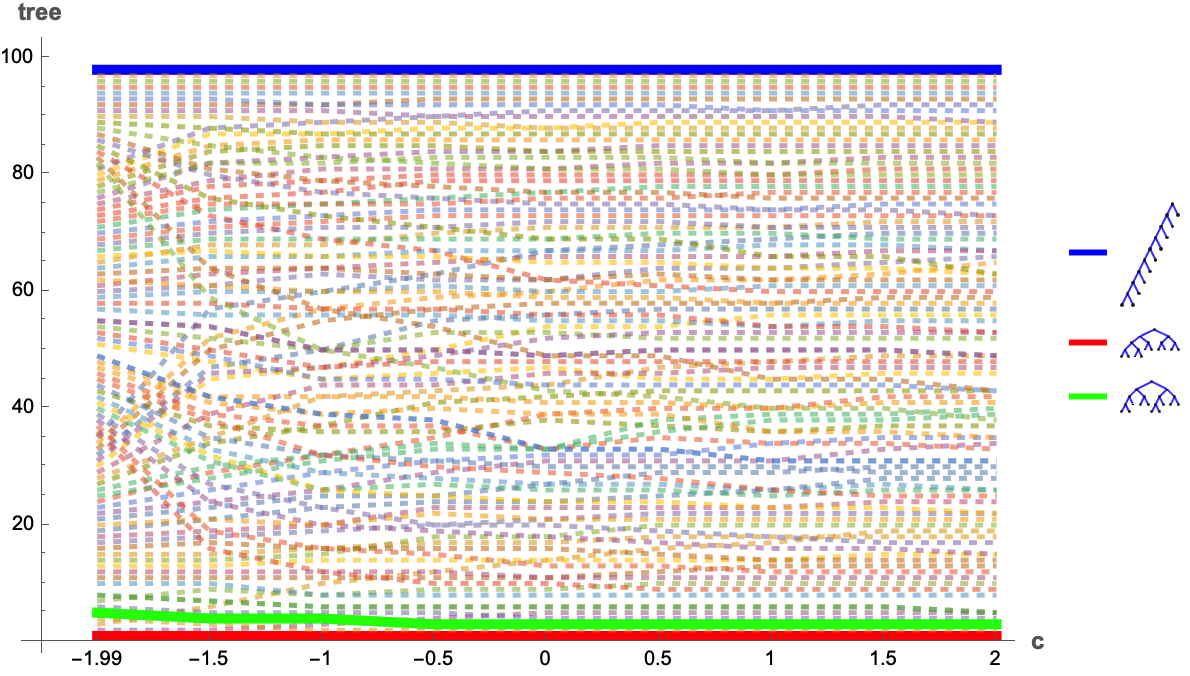}
  \caption{Ranking differences for the imbalance indices from product functions $\pi_c$ for nine different choices of $c$ with $c >-2$. The $x$-axis represents indices for values $c=-1.99, -1.5, -1, -0.5, 0, 0.5, 1, 1.5, 2$, and the $y$-axis shows the 98 rooted binary trees with 10 leaves. Each tree has an associated line showing how high (imbalanced) or low (balanced) it is ranked for the different choices of $c$, with the  rankings of three trees (caterpillar, GFB, and maximally balanced) highlighted in bold.}
  \label{fig:n10cvariable}
\end{figure}

\subsection{Explicit formulas for \texorpdfstring{$\pi_{-1}\left(T^{gfb}_n\right)$}{pi(Tgfb,-1)}  and \texorpdfstring{$\pi_0\left(T^{gfb}_n\right)$}{pi(Tgfb,0)}}\label{sec:explicitformulas} The first aim of this section is to provide two alternative direct (non-recursive) formulas to calculate the values both of $\pi_{-1}\left(T^{gfb}_n\right)$ as well as of $\pi_{0}\left(T^{gfb}_n\right)$, corresponding to the $\widehat{s}$-shape and the $Q$-shape statistic, respectively. It turns out that the sequence $\left( \pi_{-1}\left(T^{gfb}_n\right)\right)_{n \in \mathbb{N}}$ can already be found in the Online Encyclopedia of Integers Sequences OEIS \cite{OEIS}, namely under reference number A132862. However, the OEIS only contained a recursive formula to calculate $\left( \pi_{-1}\left(T^{gfb}_n\right)\right)_{n \in \mathbb{N}}$. Within the scope of the present manuscript, we have submitted our two explicit formulas for $\pi_{-1}\left(T^{gfb}_n\right)$ to the OEIS for addition to their database. Moreover,  the sequence $\pi_{0}\left(T^{gfb}_n\right)$ of the $Q$ statistic was added to the OEIS under the identifier A386912 \cite[Sequence A386912]{OEIS}.

The second aim of this section is to use our new formulas to derive explicit formulas for the minimum value of the $\widehat{s}$-shape statistic (which answers an open question \cite[Chapter 9]{fischerBook}) as well as for the $Q$-shape statistic. 

We first prove the following lemma to give insight into the greedy clustering algorithm defining $\Tgfb$.

\begin{lemma}\label{lem:AlgOnly1Non2tree}
Let $n\in \mathbb{N}$. Let $\mathcal{A}$ be the greedy clustering algorithm constructing $\Tgfb$. Then, at any point during the run of this algorithm, there is at most one tree contained in the present set of trees whose size is \emph{not} a power of two.
\end{lemma}

\begin{proof}
Towards a contradiction, we assume at some point
during the course of the algorithm, there exists more than one tree in the present set of trees whose size is \emph{not} a power of two.  Consider the first step, $j$, where two such trees exist, and call them $T$ and $T'$.  We assume that  $T = (T_1,T_2)$ (where $(T_1,T_2)$ denotes the standard decomposition of $T$ as defined above) was constructed from the trees of minimal size, $T_1$ and $T_2$ at step $i<j$. Further assume that at some later step $j$,   $T' = (T_1',T'_2)$ was constructed from the trees of minimal size, $T'_1$ and $T'_2$.  By hypothesis, at the previous step, $i-1$, all trees, have size a power of $2$.  And at step $j-1$, all trees, besides $T$, have size a power of $2$.
Since both $T_1$ and $T_2$ have size a power of $2$ and the result of joining them, $T$, does not, we know that $T_1$ and $T_2$ have different heights, and without loss of generality, we have that $|T_1| < |T_2| < |T|$.  Further, by construction, $h(T_2) + 1 = h(T)$, and, by Corollary~\ref{cor:coronado2cases}, $h(T_1) + 1 = h(T_2)$.   
By a similar argument, $|T'_1| < |T'_2| < |T'|$, $h(T'_2) + 1 = h(T')$, and $h(T'_1) + 1 = h(T'_2)$.

At each step, algorithm $\mathcal{A}$ joins the two trees of the present set of the smallest size, so, all other trees at step $i-1$ must have at least the size of $T_2$, and all other trees at step $j-1$ must have at least the size of $T'_2$.
In particular, $T$ must be as least as large as $T'_2$.  Since $T$ does not have size a power of $2$, but $T'_2$ does, we conclude from $|T|\geq |T_2'|$ that in fact $|T|> |T_2'|$. Moreover, as a power of two, $|T_2'|$ as a subtree of the GFB tree must be a fully balanced tree by Corollary \ref{cor:GFBsubtrees}, so it has minimal height for any tree with $|T_2'|$ leaves. Thus, as $|T|> |T_2'|$, we conclude that we additionally have $h(T_2)+1=h(T)>h(T_2')>h(T_1')$. This necessarily implies $h(T_2)\geq h(T_2') >h(T_1')$. Together with $h(T_2)>h(T_1)$, this gives a contradiction, as in step $i$, $T_1$ would not have been clustered with $T_2$ as tree $T_1'$ of strictly smaller size than $T_2$ was still available. This completes the proof.

\end{proof}

The following theorem will prove useful for calculating an explicit formula for the minimum value of the $\widehat{s}$-shape statistic. It counts the number of subtrees of $\Tgfb$ for all possible subtree sizes.

\begin{theorem}\label{thm:GFBsubtreesizes} Let $n \geq 1$, and let $a_n(i)$ denote the number of subtrees of $\Tgfb$ of size $i$ for $i=1,\ldots,n$. Let $k_i=\lceil \log(i) \rceil$. Then, we have: \begin{equation*} a_n(i)=\begin{cases} \left\lfloor \frac{n}{i}\right\rfloor &  \mbox{if $i=2^{k_i}$ and if $((n\mod i)=0 \ \ $ or $ \ \ (n \mod i)\geq 2^{k_i-1})$, } \\ 
\left\lfloor \frac{n}{i}\right\rfloor-1 & \mbox{if $i=2^{k_i}$ and if $(0< (n \mod i)< 2^{k_i-1})$, } \\ 
1 & \mbox{if $i \neq 2^{k_i}$ and $\left((n-i) \mod 2^{k_i-1}\right)=0$,  } \\ 0 & \mbox{if $i \neq 2^{k_i}$ and $\left((n-i) \mod 2^{k_i-1}\right)>0$. }\end{cases}.\end{equation*}

\end{theorem}

The appendix contains a proof of Theorem \ref{thm:GFBsubtreesizes}, which is based on Lemma \ref{lem:AlgOnly1Non2tree}.

We now turn our attention to $\pi_c\left(T_n^{gfb}\right)$.

\begin{corollary} \label{cor:explicitprod1} Let $n \in \mathbb{N}$, $n\geq 2$ and let $c \in \mathbb{R}, c> -2$.  Then,  $\pi_c\left(T_n^{gfb}\right) = \prod\limits_{i=2}^n (i+c)^{a_n(i)}$, where $a_n(i)$ is given by Theorem \ref{thm:GFBsubtreesizes}.
\end{corollary}

\begin{proof} By definition, we have $\pi_c(T)=\prod\limits_{v \in \mathring{V}(T)} (n_v+c)$. Since the values of $n_v$ range between 2 and $n$, sorting the subtrees by sizes and using  Theorem \ref{thm:GFBsubtreesizes} gives the desired result. 
\end{proof}

Corollary \ref{cor:sshape} 
 together with Corollary \ref{cor:explicitprod1} leads to an explicit formula for the minimum value for the $\widehat{s}$-statistic, thus solving an open problem stated in \cite[Chapter 9]{fischerBook}.

\begin{corollary}\label{cor:minvalue1}
Let $\widehat{s}^*_n=\min\limits_{T' \in \BTnstar} \widehat{s}(T')$. Then, we have: $\widehat{s}^*_n=\log \left( \prod\limits_{i=2}^n (i-1)^{a_n(i)}\right)$.
\end{corollary}

\begin{proof} By Corollary \ref{cor:sshape}, $\widehat{s}$ is uniquely minimized by $\Tgfb$, which directly implies $\widehat{s}^*_n=\widehat{s}\left(\Tgfb\right)$. By definition, this shows $\widehat{s}^*_n=\log \left( \prod\limits_{v \in \mathring{V}\left(\Tgfb\right)} (n_v-1)\right).$ Applying Corollary \ref{cor:explicitprod1} with $c=-1$ and the monotonicity of the logarithm, immediately leads to the required result and thus completes the proof.
\end{proof}

\begin{remark} The sequence $\left( \prod\limits_{i=2}^n (i-1)^{a_n(i)}\right)_{n \geq 1}$ starts with the values:
$$1, 1, 2, 3, 8, 15, 36, 63, 192, 405, 1080, 2079, 6048, 12285, 31752, 59535
193536, 433755, 1224720$$ 
has already occurred in different contexts (cf. \cite{bodini}). It is listed in the Online Encyclopedia of Integer Sequences (\cite[Sequence A132862]{OEIS}), where only recursive formulas were given.  Using the GFB tree, in particular Theorem \ref{thm:GFBsubtreesizes} and Corollary~\ref{cor:minvalue1}, now allows for a simple explicit formula to calculate this sequence.\footnote{We note that another explicit, but more complicated, formula was independently derived in \cite{theresa}.} The explicit formulas arising from the present manuscript have meanwhile been submitted to the OEIS. 
\end{remark}

Analogously, we can derive a result as given by Corollary~\ref{cor:minvalue1} also for the $Q$-shape statistic.

\begin{corollary}\label{cor:minvalue1Q}
Let $Q^*_n=\min\limits_{T' \in \BTnstar} Q(T')$. Then, we have: $Q^*_n=\log \left( \prod\limits_{i=2}^n i^{a_n(i)}\right)$.
\end{corollary}

\begin{proof} Using $c=0$, the proof of Corollary \ref{cor:minvalue1Q} is analogous to the proof of Corollary \ref{cor:minvalue1}.
\end{proof}

While the explicit formulas for $\pi_{-1}\left(T^{gfb}_n\right)$ and for $\pi_{0}\left(T^{gfb}_n\right)$ provided by Corollary~\ref{cor:explicitprod1} are non-recursive and thus already an improvement to the previous state of the literature, they heavily depend on Theorem~\ref{thm:GFBsubtreesizes}. The following theorem provides another direct formula independent of the values $a_n(i)$. 

\begin{theorem} \label{thm:explicitprod2}
Let $n \in \mathbb{N}$, $n\geq 2$ and let $c \in \mathbb{R}, c> -2$. Let $k_n=\lceil \log(n)\rceil$,  and let $d_n=
 n-2^{k_n-1}$ be the difference between $n$ and the next lower power of 2. Then, we have:

\begin{enumerate}
\item If $d_n=2^{k_n-1}$ (with $n=2^{k_n}$), then $\pi_c\left(T_n^{gfb}\right) =\pi_c\left(T_{k_n}^{fb}\right) = \prod\limits_{i=0}^{k_n-1} (2^{k_n-i}+c)^{2^i}$.
\item If $d_n<2^{k_n-1}$, (with $n<2^{k_n}$), then we have:  \begin{align*}\pi_c\left(T_n^{gfb}\right) = & \pi_c\left(T_{k_n-1}^{fb}\right)\cdot (2+c)^{d_n} \\ &\cdot \prod\limits_{i=1}^{k_n-1}\left( \left(\frac{2^{i+1}+c}{2^i+c}\right)^{\left\lfloor \frac{d_n}{2^i} \right\rfloor} \cdot \left(\frac{2^i+\left(d_n-2^i\cdot\left\lfloor\frac{d_n}{2^i}\right\rfloor\right)+c}{2^i+c}\right)^{\left\lceil \frac{d_n}{2^i}\right\rceil-\left\lfloor\frac{d_n}{2^i}\right\rfloor}\right)\end{align*}

\end{enumerate}
\end{theorem}

Before we can prove Theorem \ref{thm:explicitprod2}, we briefly show that the complete tree $T_n^c$  of Fill  \cite{Fill1996} coincides with the GFB tree $\Tgfb$, which has been independently shown in Riesterer \cite[Theorem 4.2]{theresa}, albeit with a different proof. This result will enable us to use a different construction of $T_n^{gfb}$, namely to start with $T_{k_n-1}^{fb}$, fix an orientation and then replace a certain number of leaves by cherries from left to right, until $n$ leaves in total are reached.

\begin{lemma}[Theorem 4.2 in  \cite{theresa}]\label{lem:Tgfb_altDef}
Let $n \geq 1$. Then, $\Tgfb=T^c_n$.  
\end{lemma}

\begin{proof} We prove the statement by induction on $n$. In the following, let $\Tgfb=(T_a,T_b)$, where $T_a$ has size $n_a$ and $T_b$ has size $n_b$ with $n_a \geq n_b$. Similarly, let $T_n^c=(T_a^c,T_b^c)$, where $T_a^c$ has size $n_a^c$ and $T_b^c$ has size $n_b^c$ with $n_a^c \geq n_b^c$.

For $n= 1$ there is only one tree, so there is nothing to show. We now assume the statement holds for all trees with up to $n-1$ leaves and consider $n \geq 2$. Let $\ell_{n} =\lfloor \log(n) \rfloor$, that is $n \in \{2^{\ell_{n}},\ldots,2^{\ell_{n}+1}-1\}$. Then, by construction of $T_n^c$, we know that if $n\leq 3 \cdot 2^{\ell_{n}-1}$, $n_b^c=2^{\ell_{n}-1}$ (as in this case, no leaves of $T_b$ get replaced by cherries), and $n_a^c=n-n_b^c=n-2^{\ell_{n}-1}$. By Part (1) of Proposition \ref{prop:coronado2cases}, we can then conclude $n_a=n_a^c$ and $n_b=n_b^c$. 

If, on the other hand,  $n>3 \cdot 2^{\ell_{n}-1}$, we know by construction that $T_a^c=T_{\ell_{n}}^{fb}$ and thus $n_a^c=2^{\ell_{n}}$ (as \emph{all} leaves of the left subtree have been replaced by cherries) and thus $n_b^c=n-n_a^c=n-2^{\ell_{n}}$. By Part (2) of Proposition \ref{prop:coronado2cases}, we can then conclude $n_a=n_a^c$ and $n_b=n_b^c$. 

Thus, in both cases, the sizes of $T_a$ and $T_a^c$ as well as the sizes of $T_b$ and $T_b^c$, respectively, coincide, which by induction shows that $T_a=T_a^c$ and $T_b=T_b^c$. Thus, $\Tgfb=(T_a,T_b)=(T_a^c,T_b^c)=T_n^c$, which completes the proof.
\end{proof}

We are now in the position to prove Theorem \ref{thm:explicitprod2}. 
\begin{proof}[Proof of Theorem \ref{thm:explicitprod2}] For the case in which $d_n=2^{k_n-1}$, the statement is a direct conclusion of the fact that for $n=2^{k_n}$, the GFB tree and the fully balanced tree coincide.

It remains to consider the case $d_n<2^{k_n-1}$. In this case, by Lemma \ref{lem:Tgfb_altDef}, we can derive the GFB tree from $T_{k_n-1}^{fb}$ by replacing $d_n$ leaves by cherries from left to right. Thus, we can start with $\pi_c\left(T_{k_n-1}^{fb}\right)$ and modify this product by dividing out all factors of $\pi_c\left(T_{k_n-1}^{fb}\right)$ that are no longer present in $\pi_c\left(T_{n}^{gfb}\right)$ and multiypling in factors newly occuring in the latter term. We do this using the following considerations: 
\begin{itemize}
\item The factor $(2+c)^{d_n}$ stems from the $d_n$ newly formed cherries, each of which contributes a factor of $2+c$. For this new factor, nothing needs to be divided out as the former leaves did not occur in the term $\pi_c\left(T_{k_n-1}^{fb}\right)$. 
\item Next, for each $i=1,\ldots,k_n-1$, we need to check how many subtrees of size $2^i$ in $T_{k_n-1}^{fb}$ get replaced by subtrees of size $2^{i+1}$ when $T_n^{gfb}$ is formed. These are precisely the subtrees of size $2^i$ in $T_{k_n-1}^{fb}$ all of whose leaves get replaced by cherries. As we fill the tree up from left to right, it can be easily seen that there are $\left\lfloor \frac{d_n}{2^i} \right\rfloor$ such trees, explaining the term $\left(\frac{2^{i+1}+c}{2^i+c}\right)^{\left\lfloor \frac{d_n}{2^i} \right\rfloor}$.
\item Last, for each $i=1,\ldots,k_n-1$, there may be at most one subtree of size $2^i$ in $T_{k_n-1}^{fb}$ some of whose leaves but not all get replaced by cherries. This depends on whether $\frac{d_n}{2^i}$ is an integer. If it is, then the $d_n$ new cherries completely fill up all trees of size $2^i$ to which they were added, (that is, all these trees have already been considered in the previous term). This is the case when $\left\lceil \frac{d_n}{2^i}\right\rceil-\left\lfloor\frac{d_n}{2^i}\right\rfloor=0$, which implies that in this case, the latter term in the equation equals 1 and thus does not contribute to the overall product.

If, however, $\frac{d_n}{2^i}$ is not an integer and thus $\left\lceil \frac{d_n}{2^i}\right\rceil-\left\lfloor\frac{d_n}{2^i}\right\rfloor=1$, then there is a tree of size $2^i$ to which $d_n-2^i \cdot \left\lfloor\frac{d_n}{2^i} \right\rfloor$ leaves get added, namely precisely the \enquote{left over} leaves after filling up $\left\lfloor\frac{d_n}{2^i} \right\rfloor$ many subtrees of $T_{k_n-1}^{fb}$ with $2^i$ new leaves each. This explains the last factor and thus completes the proof.
\end{itemize}
\end{proof}

Next, we again turn our attention to the $\widehat{s}$-shape statistic.

\subsection{Expected values of \texorpdfstring{$\widehat{s}$}{s}-shape statistic}\label{sec:expected} The $\widehat{s}$-shape statistic plays an important role in tree balance, particularly in the context of mathematical phylogenetics and the Yule-Harding model, cf. \cite{kersting2024}. However, so far the expected values of the $\widehat{s}$-shape statistic under the Yule-Harding and the uniform models, which are common distributions of trees, are unknown \cite[Chapter 9]{fischerBook}. In the following, we give bounds on the expected value of the $\widehat{s}$-shape statistic under these two distributions: the uniform distribution where each (binary) tree on $n$ leaves is equally likely, and the Yule-Harding distribution. To show our bounds, we outline and use the elegant approach of King and Rosenberg \cite{king2021}.  They note that the expectation of the Sackin index can be computed in terms of the cluster sizes for any distribution, $\theta$, on trees which has the exchangeability property, that is, for each $T\in \mathcal{BT}_n$ and each permutation, $\sigma$, of its leaf labels, $\mathbb{P}_{\theta}(T) = \mathbb{P}_{\theta}(\sigma(T))$.  For such distributions, the proposition from Than and Rosenberg \cite{than2014} applies for the Sackin index, $S_n$:

\begin{proposition} [\cite{than2014}, Lemma 6]  If a probability distribution, $\Theta$, on $\mathcal{BT}_n$ satisfies the exchangeability property, then the expected value for the Sackin index on $n$-leaf trees is:
$$
    \mathbb{E}_{\theta}[S_n] =  \sum_{k=1}^{n-1} \binom{n}{k} kp_n(k)
$$
where $p_n(k)$ is the probability that a given subset $A \subseteq X$ with $|A| = k$, $1\leq k\leq n$, is a cluster of a tree of size $n$ leaves sampled from $\mathcal{BT}_n$ according to $\Theta$.
\end{proposition}

Note that the $\widehat{s}$-shape statistic, like the Sackin index, sums across all cluster sizes of a given tree: the Sackin index sums up the respective size $k$, while the $\widehat{s}$-statistic sums the logarithm of said size (namely, $\log(k-1)$). This similarity between the indices allows us to use the above approach introduced for the Sackin index also for the $\widehat{s}$-shape statistic:

\begin{proposition}\label{prop:exp_value}  If a probability distribution, $\Theta$, on $\mathcal{BT}_n$ satisfies the exchangeability property, then the expected value for the $\widehat{s}$-shape statistic on $n$-leaf trees is
$$
    \mathbb{E}_{\theta}[\widehat{s}_n] =  \sum_{k=2}^{n-1} \binom{n}{k} \log(k-1)p_n(k)
$$
where $p_n(k)$ is the probability that a given subset $A \subseteq X$ with $|A| = k$, $2\leq k\leq n$, is a cluster of a tree of size $n$ leaves sampled from $\mathcal{BT}_n$ according to $\Theta$.
\end{proposition}

King and Rosenberg \cite{king2021} give an elegant proof of the expected value of the Sackin index and the resulting closed form: 

\begin{theorem}[\cite{king2021}, Corollary 7]
The expectation of the Sackin index under the uniform model on rooted binary labeled trees of $n$ leaves is:
    $$
        \mathbb{E}_{U}[S_n] = \frac{4^{n-1}}{C_{n-1}} - n
    $$
where $C_{n} = \frac{1}{n+1}\binom{2n}{n}$, the $n^{th}$ Catalan number.
\end{theorem} 

Using the bounds from \cite{king2021}, we can show:

\begin{theorem}
The expectation of the $\widehat{s}$-shape statistic under the uniform model on rooted binary labeled trees of $n$ leaves is:
    $$
        \frac{\log n}{n}\left[\frac{4^{n-1}}{C_{n-1}} - n\right]
        \leq \mathbb{E}_{U}[\widehat{s}_n] \leq
        \frac{4^{n-1}}{C_{n-1}} - n
    $$
\end{theorem} 

\begin{proof}  The upper bound follows directly from the Sackin index being an upper bound for the $\widehat{s}$-shape statistic.

For the lower bound, we use Proposition \ref{prop:exp_value}  and the simple bound of $\frac{\log(k-1)}{k} \geq \frac{\log(n-1)}{n}$ for $2\leq k \leq n$:
$$
\begin{array}{rcl}
    \mathbb{E}_{\theta}[\widehat{s}_n] & = & \sum_{k=2}^{n-1} \binom{n}{k} \log(k-1)p_n(k) 
     =  \sum_{k=2}^{n-1} \binom{n}{k} \frac{\binom{n-1}{k-1}}{\binom{2n-2}{2k-2}}\log(k-1) \\
    & \geq & 
    \frac{\log (n-1)}{n C_{n-1}}(4^{n-1}
    -[\binom{2(n-1)}{n-1} + 2\binom{2n-5}{n-3} + \binom{2(n-1)}{n-1}])\\   
    & = & \frac{\log (n-1)}{n C_{n-1}}(4^{n-1}
    -[2\binom{2(n-1)}{n-1} + 2\binom{2n-5}{n-3}]).\\     
\end{array}
$$
This completes the proof.
\end{proof}

\section{Discussion and Conclusion}\label{sec:discussion} We have introduced two families of imbalance indices, namely $\Phi_f$ for strictly increasing and strictly concave functions $f$ and the product functions $\pi_c$ for $c>-2$, and  shown that both of them are uniquely minimized by the GFB tree. For the $\widehat{s}$-shape statistic, which is an important imbalance index used in the phylogenetic literature, this finding answered the open question concerning its extrema from \cite[Chapter 9]{fischerBook}. However, our approach is more general and not just focused on the $\widehat{s}$-shape statistic. In particular, we have shown that our families of imbalance indices contain infinitely many different imbalance indices, some of which might be useful in phylogenetics and other research areas where tree balance plays a role.

\section*{Acknowledgements}
The authors wish to thank Tom N. Hamann, Kristina Wicke and Volkmar Liebscher for helpful discussions. The present manuscript is based upon work supported by the National Science Foundation under Grant No. DMS-1929284 while the authors were in residence at the Institute for Computational and Experimental Research in Mathematics (ICERM) in Providence, RI, during the Theory, Methods, and Applications of Quantitative Phylogenomics semester program.

\bibliographystyle{plain}
\bibliography{trees.bib}

\section*{Declarations}

\begin{itemize}
\item \textbf{Funding:} The present manuscript is based upon work supported by the National Science Foundation under Grant No. DMS-1929284 while the authors were in residence at the Institute for Computational and Experimental Research in Mathematics (ICERM) in Providence, RI, during the Theory, Methods, and Applications of Quantitative Phylogenomics semester program.
\item \textbf{Competing interests:} All others received the funding declared above. The authors confirm that they have no other financial or non-financial interests to declare.
\item \textbf{Data availability:} The authors declare that all data generated for this study are contained in the manuscript. No other sources of data have been used.
\end{itemize}

\noindent

\section*{APPENDIX}

Here we present the somewhat technical proof of Theorem \ref{thm:GFBsubtreesizes}.

\begin{proof}[Proof of Theorem \ref{thm:GFBsubtreesizes}] 

We first analyze the procedure with which the GFB tree is generated. When the greedy clustering is performed to form $\Tgfb$, we refer to all clusterings that involve at least one tree of size $2^{m}$ and no tree of size $2^{m-1}$ as phase $m$ of the algorithm. Phase 0 includes all steps that cluster subtrees of size $2^m=2^0=1$, that is the leaves to form cherries.  If $n$ is odd, this includes as the last step the clustering of a leaf and a cherry (subtree of size 2).  Phase 1 contains all steps that involve clustering of cherries.  Almost all these clusterings will be two cherries clustered into a new subtrees of size 4.  The last clustering in phase 1 could cluster a cherry with a subtree of size 3 or 4, if $n \mod 4 \equiv 3$ or $2$, respectively. 

We now proceed by considering some $i \in \{1,\ldots,n\}$ with $k_i=\lceil \log(i) \rceil$  as the end of phase $k_i-1$ of the algorithm (or, equivalently, the beginning of phase $k_i$). The trees present at this stage will be referred to as the \emph{current set} of trees. All trees in the current set must have size at least $2^{k_i-1}+1$ as all trees of size up to $2^{k_i-1}$ have already been clustered into larger trees. Moreover,  all trees in the current set have size strictly smaller than $2^{k_i+1}$. 
This is due to the following: first, phase $k_i-1$ only performs clusterings in which at least one tree has size at most $2^{k_i-1}$. Second, all trees formed throughout the course of the algorithm are again GFB trees by Corollary \ref{cor:GFBsubtrees}. This implies by Corollary \ref{cor:coronado2cases} that a tree of size $2^{k_i+1}$ or larger would have two maximum pending subtrees of size at least $2^{k_i}$ each -- but two such trees could not have been clustered during phase $k_i-1$ of the algorithm. 
So, all trees $T$ in the current set have size $n_T \in S=\{2^{k_i-1}+1,\ldots,2^{k_i+1}-1\}$. However, by Lemma \ref{lem:AlgOnly1Non2tree}, there is at most one tree present in the current set of trees whose size is not a power of 2. As $S$ only contains one power of 2, namely $2^{k_i}$, we conclude that all trees except possibly for one have size $2^{k_i}$.

We now consider the decomposition  $n=a\cdot 2^{k_i}+b$ with $b=\left(n \mod 2^{k_i}\right) \in \{0,\ldots,2^{k_i}-1\}$ and $a=\left\lfloor \frac{n}{2^{k_i}} \right\rfloor$. 
We distinguish three cases depending on the value of $b$. Note that this decomposition together with the above observations immediately implies that there are $a-1$ or $a$ trees of size $2^{k_i}$ in the current set. In the first case, the only remaining tree has size $2^{k_i}+b$, and in the second case, the only remaining tree has size $b$. 

\begin{itemize}
\item If $b=0$, we have $n=a\cdot 2^{k_i}$. 
Thus, we have $(n \mod 2^{k_i})=0$, and obviously all  trees in the current set have size $2^{k_i}$, because $n-(a-1)\cdot 2^{k_i} = 2^{k_i}$, so if you consider $a-1$ trees of size $2^{k_i}$ and one remaining tree, the remaining tree has the same size. This directly proves the statement $a_n(i)=\left \lfloor \frac{n}{i}\right\rfloor$ of the theorem for $i=2^{k_i}$ and $(n \mod i)=0$ (as all trees of the current set are subtrees of $\Tgfb$ and no other trees of size $i=2^{k_i}$ can be formed throughout the algorithm). 

Moreover, as $b=0$, we also have for $i \in \{2^{k_i-1}+1,\ldots,2^{k_i}-1\}$ that $((n-i)\mod 2^{k_i-1})>0$, and as there is no tree of such a size $i$ (as all trees in the current set have size $2^{k_i}$, which by Corollary \ref{cor:coronado2cases} are formed of two trees of size $2^{k_i-1}$ each), we have $a_n(i)=0$ for such values of $i$. This is in accordance with the last case of the theorem.

\item If $0<b<2^{k_i-1}$, we observe that there cannot be $a$ trees of size $2^{k_i}$ in the current set. Otherwise, there would be $b$ leaves not contained in any $2^{k_i}$ tree, thus belonging to a tree $T$ of size $b$ in the current set -- but this would imply that the size of $T$ is strictly smaller than $2^{k_i-1}$, a contradiction to the current set being the set at the \emph{end} of phase $k_i-1$, as $T$ would have had to cluster in this phase. 

So the only possibility is that there are $a-1$ trees of size $2^{k_i}$ in the current set and one tree $T$ of size $2^{k_i}+b$. But by Corollary \ref{cor:coronado2cases}, $T$ cannot have a subtree of size $2^{k_i}$ (as a tree of size $2^{k_i}$ cannot be clustered with a tree of size $b<2^{k_i-1}$ to form a GFB tree, and $T$ has to be a GFB tree, too, as it is a subtree of $\Tgfb$); it instead has one maximum pending subtree of size $2^{k_i-1}$ and one of size $2^{k_i-1}+b$. 

Thus, $\Tgfb$ contains $a-1=\left\lfloor \frac{n}{i}\right\rfloor-1$ trees of size $i=2^{k_i}$, so $a_n\left(2^{k_i}\right)=\left\lfloor \frac{n}{i}\right\rfloor-1$, which proves the second case of the theorem as we have $0<b=(n \mod i)<2^{k_i-1}$. 

Further, if $i=2^{k_i-1}+b$, we have $n-i=a\cdot 2^{k_i}+b-\left(2^{k_i-1}+b \right)=(2a-1)\cdot 2^{k_i-1}$, and thus $((n-i)\mod 2^{k_i-1})=0$. As we have already seen that a subtree of size 
$2^{k_i-1}+b$ was formed in the course of the algorithm (and no more such trees can be formed as the algorithm proceeds), we have $a_n\left(2^{k_i-1}+b\right)=1$, which is in accordance with the third case of the theorem. 

Last, if $i\in \{2^{k_i-1}+1,\ldots,2^{k_i}-1\}$ with $i \neq 2^{k_i-1}+b$, we have seen that no such tree can be formed in the course of the algorithm, so $a_n(i)=0$. As in this case we have $((n-i)\mod 2^{k_i-1})>0$, this is in accordance with the fourth case of the theorem.

\item Finally, we consider the case $b>2^{k_i-1}$. In this case, before considering the beginning of phase $k_i$ of the algorithm, we consider the beginning of phase $k_i-1$, which is the end of phase $k_i-2$. We can write $n=a\cdot 2^{k_i}+b=(2a+1)\cdot 2^{k_i-1}+b'$, where $b'=b-2^{k_i-1}<2^{k_i-1}$ (as $b<2^{k_i}$). Using the same reasoning as above, in the beginning of phase $k_i-1$, there must either be $2a$ trees of size $2^{k_i-1}$ and one tree of size $n-2a\cdot 2^{k_i-1}=2^{k_i-1}+b'$, or there are $2a+1$ trees of size $2^{k_i-1}$ and one tree of size $b'<2^{k_i-1}$. 

In the latter case, if in the beginning of phase $k_i-1$ we have $2a+1$ trees of size $2^{k_i-1}$ and one tree of size $b'<2^{k_i-1}$, note that we must have $b'>0$ (empty trees are never formed in the algorithm). Then, the tree of size $b'$ will be the first one to be clustered in phase $k_i-1$ to one of the $2a+1$ trees of size $2^{k_i-1}$. This forms a tree of size $2^{k_i-1}+b'$, and the remaining $2a$ trees of size $2^{k_i-1}$ will cluster to form $a$ trees of size $2^{k_i}$ in the course of phase $k_i-1$. Thus, in the end of this phase, we have $a$ trees of size $2^{k_i}$ and one tree of size $2^{k_i-1}+b'=2^{k_i-1}+b-2^{k_i-1}=b$. This tree of size $b$ would be the first tree  to be clustered in phase $k_i$ with one of trees of size $2^{k_i}$ to form a tree of size $2^{k_i}+b$.

Now, we consider the first case,  the case where in the beginning of phase $k_i-1$ we have $2a$ trees of size $2^{k_i-1}$ and one tree of size $2^{k_i-1}+b'$. If $b'>0$, phase $k_i-1$ would not involve the last tree, it would only cluster the $2a$ trees of size $2^{k_i-1}$ to form $a$ trees of size $2^{k_i}$. The remaining tree of size $2^{k_i-1}+b'$ would be the first one to be clustered in phase $k_i$ with one of trees of size $2^{k_i}$, though, to form a tree of size $3\cdot 2^{k_i-1}+b'= 2^{k_i}+b'+2^{k_i-1}=2^{k_i}+b$. Similarly, if $b'=0$, phase $k_i-1$ would first cluster the $2a$ trees of size $2^{k_i-1}$ to form $a$ trees of size $2^{k_i}$ and then, in its last step, it would cluster the remaining tree of size $2^{k_i-1}+b'=2^{k_i-1}$ with one of the already formed trees of size $2^{k_i}$ to form a tree of size $3\cdot 2^{k_i-1}+b'=2^{k_i}+b$. 

So in all cases, we can see (either directly after phase $k_i-1$ or after the first step of phase $k_i$) that $\Tgfb$ contains precisely $a$ trees of size $2^{k_i}$, one of which gets clustered with a tree of size  $b$ to form a tree of size $2^{k_i}+b$. 

Now, if $i=2^{k_i}$, the $a=\left\lfloor \frac{n}{i} \right\rfloor$ many trees of size $2^{k_i}$ are in accordance with the first case of the theorem, because here we have $(n \mod 2^{k_i}) = b \geq 2^{k_i-1}$. Therefore, in this case we have $a_n(i)=\left\lfloor \frac{n}{i} \right\rfloor$. 

However, if $i\in \{2^{k_i-1}+1,\ldots,2^{k_i}-1\}$, we have $a_n(i)=1$ if $i=b$ (because a tree of size $b$ gets clustered with a tree of size $2^{k_i}$). This is in accordance with the third case of the theorem, as here we have $n-i=n-b = a\cdot 2^{k_i}+b-b=a\cdot 2^{k_i}$, and thus $((n-i)\mod 2^{k_i-1})=0$. If, on the other hand, we have $i\in \{2^{k_i-1}+1,\ldots,2^{k_i}-1\}$ and $i\neq b$, then there cannot be a subtree of size $i$: It would have to be clustered away last in phase $k_i$ of the algorithm. So it would have to be there in the end of phase $k_i-1$ (which is not the case as we have analyzed above), or it would have to arise within phase $k_i$. The latter cannot happen as we have seen that at the latest after one step of phase $k_i$, we are left with only trees of size at least $2^{k_i}$. Thus, we conclude that for such values of $i$, we have $a_n(i)=0$. This is in accordance with the fourth case of the theorem, because in this case we have $n-i=a\cdot 2^{k_i}+b-i$. Since both $b$ and $i$ are contained in the set $\{2^{k_i-1},\ldots,2^{k_i}-1\}$, their maximal difference is bounded by $2^{k_i}-2^{k_i-1}$.
We can distinguish two cases: If $b>i$, we have $n-i=a\cdot 2^{k_i}+b-i>a\cdot 2^{k_i}=2a\cdot 2^{k_i-1}$ as well as $n-i=a\cdot 2^{k_i}+b-i< a\cdot 2^{k_i} + 2^{k_i}-2^{k_i-1}=(2a+1)\cdot 2^{k_i-1}$, so we have $(2a+1)\cdot 2^{k_i-1}>n-i>2a\cdot 2^{k_i-1}$, which shows that $((n-i)\mod 2^{k_i-1})>0$. Similarly, if $b<i$, we have $n-i=a\cdot 2^{k_i}+b-i<a\cdot 2^{k_i}=2a\cdot 2^{k_i-1}$ as well as $n-i=a\cdot 2^{k_i}+b-i> a\cdot 2^{k_i} - 2^{k_i}+2^{k_i-1}=(2a-1)\cdot 2^{k_i-1}$, so we have $2a\cdot 2^{k_i-1}>n-i>(2a-1)\cdot 2^{k_i-1}$, which shows that $((n-i)\mod 2^{k_i-1})>0$. 
\end{itemize}
This completes the proof of the theorem.

\end{proof}

\end{document}